\documentclass[draftcls, onecolumn]{IEEEtran}
\usepackage{amsmath,amsthm,amssymb,paralist,subfigure,cite,graphicx,algorithm,url,tikz}
\newtheorem{lem}{Lemma}
\newtheorem{thm}{Theorem}

\newtheorem{ex}{Example}

\newtheorem{cor}{Corollary}
\def\ve{\varepsilon}
\def\mt{\mathbf{x}}

\def\mb{\mathbf}
\def\mbb{\mathbb}
\def\mc{\mathcal}
\def\wh{\widehat}
\def\wt{\widetilde}
\def\ra{\rightarrow}
\def\ol{\overline}
\def\ul{\underline}

\def\rank{\mbox{\upshape{rank}}}

\title{Consensus in the presence of interference}
\author{Usman~A.~Khan,~\IEEEmembership{Member, IEEE}, Shuchin~Aeron,~\IEEEmembership{Member, IEEE}
\thanks{The authors are with Department of Electrical and Computer Engineering at Tufts University, Medford, Email: \texttt{\{khan,shuchin\}@ece.tufts.edu}.
}}

\begin{document}
\maketitle

\begin{abstract}
This paper studies distributed strategies for average-consensus of arbitrary vectors in the presence of network interference. We assume that the underlying communication on any \emph{link} suffers from \emph{additive interference} caused due to the communication by other agents following their own consensus protocol. Additionally, no agent knows how many or which agents are interfering with its communication. Clearly, the standard consensus protocol does not remain applicable in such scenarios. In this paper, we cast an algebraic structure over the interference and show that the standard protocol can be modified such that the average is reachable in a subspace whose dimension is complimentary to the maximal dimension of the interference subspaces (over all of the communication links). To develop the results, we use \emph{information alignment} to align the intended transmission (over each link) to the null-space of the interference (on that link). We show that this alignment is indeed invertible, i.e. the intended transmission can be recovered over which, subsequently, consensus protocol is implemented. That \emph{local} protocols exist even when the collection of the interference subspaces span the entire vector space is somewhat surprising. 
\end{abstract}

\section{Introduction}

In this paper, we consider the design and analysis of average-consensus protocols (averaging vectors in $\mbb{R}^n$) in the presence of network interference. Each agent, while communicating locally with its neighbors for consensus, causes an interference in other communication links. We assume that these interferences are additive and lie on low-dimensional subspaces. Such interference models have been widely used in several applications, e.g. electromagnetic brain imaging~\cite{Anonymous:HlXifkm4}, magnetoencephalography~\cite{Gutierrez:2004ca,Sekihara:2004bn}, beamforming~\cite{McCloud:1998im,Dogandzic:2002jm}, and multiple-access channels~\cite{Lupas:1989bz,Varanasi:1998dx}. Interference cancellation, thus, has been an important subject of study in the aforementioned areas towards designing matched detectors, adaptive beamformers, and generalized hypothesis testing~\cite{Scharf:1994jv,McCloud:1997je,Goldstein:1997jb,Wang:1998cm,Dogandzic:2007dj,Monsees:2013if}. 

As distributed architectures are getting traction, information is to be distributedly processed for the purposes of learning, inference, and actuation. Average-consensus, thus, is a fundamental notion in distributed decision-making, see~\cite{jadbabailinmorse03,Xiao05distributedaverage,Mesbahi-parameter,Giannakis-est,usman_hdctsp09,olfati:cdc09,Sayed-LMS,KarSIAM2013} among others. When the inter-agent communication is noiseless and interference-free, the standard protocol is developed in~\cite{boyd:04}. Subsequently, a number of papers~\cite{Bucklew,Kashyap,Tuncer} consider average-consensus in imperfect scenarios. Reference~\cite{Kar_TSP2009} considers consensus with link failures and channel noise, while~\cite{Chen_ECDC2011} addresses asymmetric links with asymmetry in packet losses. Consensus under stochastic disturbances is considered in~\cite{Aysal_TIT2010}, while~\cite{Nazer_JSTSP2011} studies a natural superposition property of the communication medium and uses computation codes to achieve energy efficient consensus. 


In contrast to the past work outlined above, we focus on an algebraic model for network interference. We assume that the underlying communication on any link suffers from additive interference caused due to the communication by other agents following their own consensus protocol. The corresponding interference subspace, in general, depends on the communication link and the interfering agent. A fortiori, it is clear that if the interference by an agent is persistent in all dimensions~($\mbb{R}^n$), there is no way to recover the true average unless schemes similar to interference alignment~\cite{Jafar11} are used. In these alignment schemes, the data is projected onto higher dimensions such that the interferences and the data lie in different low-dimensional subspaces; clearly, requiring an increase in the communication resources. 

On the other hand, if the interference from each agent already lies in (possibly different) low-dimensional subspaces, the problem we address is whether one can exploit this low-dimensionality for interference cancellation, and subsequently, for consensus. Furthermore, we address how much information can be recovered when the collection of the local interferences span the entire vector space, $\mbb{R}^n$? Our contribution in this context is to develop information alignment strategies for interference cancellation and derive a class of (vector) consensus protocols that lead to a meaningful consensus. In particular, we show that the prospoed alignment achieves the average in a subspace whose dimension is complimentary to the maximal dimension of the interference subspaces (over all of the communication links).

To be specific, if agent~$j$ sends~$\mb{x}^j\in\mbb{R}^n$ to agent~$i$, agent~$i$ actually receives\footnote{In general, the interference matrix, $\Gamma$, may depend on the particular link,~$j\ra i$, and the interfering agent,~$m$, and will be denoted by~$\Gamma_{ij}^m$.}~$\mb{x}^j + \sum_m\Gamma\mb{x}^m$, with~$\ol{\gamma}\triangleq\rank(\Gamma)<n$. 
In this context, we address the following challenges: 
\begin{inparaenum}[(i)]
\item The received signal is corrupted by several interferers, each on a distinct (low-rank) subspace. Is it possible to design a \emph{local}  operation that cancels each interference? 
\item The aforementioned cancellation has to be \emph{locally} reversible (to be elaborated later) in order to build a meaningful consensus. 
\item The signal hampered with interference passes through consensus weights,~$w_{ij}$, iteratively. Notice again the received signal,~$\sum_{j\in\mc{N}_i}w_{ij}(\mb{x}^j + \sum_m\Gamma\mb{x}^m)$, at agent~$i$, where~$\mc{N}_i$ is the neighbors at agent~$i$. An arbitrary small disturbance due to the interference can result in perturbing the spectral radius of the consensus weight matrix to~$1+\ve$, which forces the iterations to converge to~$0$ when~$\ve<0$, or diverge when~$\ve>0$,~\cite{usman_hdctsp09}.
\end{inparaenum}

We explicitly assume that no agent in the network knows how many and which agents may be interfering with its received signals. Additionally, we assume that only the null space of the underlying interferences are known \emph{locally} (singular values and basis vectors may not be known). Within these assumptions, it is clear that the aforementioned challenges are non-trivial. What we describe in this paper are completely local \emph{information alignment} strategies that not only ensure that average-consensus is reached, but also characterize where this consensus is reached. In particular, we show that average of the initial conditions, vectors in~$\mbb{R}^n$, can be recovered in the subspace whose dimension, $n-\ol{\gamma}$, is complimentary to the (maximal) dimension, $\ol{\gamma}$, of the local interferences.

The rest of the paper is organized as follows. Section \ref{pre_not} outlines the notation and gathers some useful facts from linear algebra. Section \ref{pf} formulates the problem while Section \ref{aca} presents a simple architecture, termed as \emph{uniform interference}, and develops the information alignment scheme. Section \ref{aca} then identifies two generalizations of the uniform interference, namely \emph{uniform outgoing interference} and \emph{uniform incoming interference}, subsequently treated in Sections \ref{s_uoi} and \ref{s_uii}, respectively. In each of these sections, we provide simulations to illustrate the main theoretical results and their implications. Section \ref{s_discuss} provides a summary and discussion of the main results and Section~\ref{s_conclude} concludes the paper.

\section{Notation and Preliminaries}\label{pre_not}
We use lowercase bold letters to denote vectors and uppercase italics for matrices (unless clear from the context). The symbols~$\mb{1}_n$ and~$\mb{0}_n$ are the~$n$-dimensional column vectors of all~$1$'s and all~$0$'s, respectively. The identity and zero matrices of size~$n$ are denoted by~$I_n$ and~$\mb{0}_{n\times n}$, respectively. We assume a network of~$N$ agents indexed by,~$i=1,\ldots,N$, connected via an undirected graph,~$\mc{G}=(\mc{V},\mc{E})$, where~$\mc{V}$ is the set of agents, and~$\mc{E}$, is the set of links,~$(i,j)$, such that agent~$j\in\mc{V}$ can send information to agent~$i\in\mc{V}$, i.e.~$j\ra i$. Over this graph, we denote the neighbors of agent~$i$ as~$\mc{N}_i$, i.e. the set of all agents that can send information to agent~$i$: $\mc{N}_i = \{j~|~(i,j)\in\mc{E}\}.$

In the entire paper, the initial condition at an agent,~$i\in\mc{V}$, is denoted by an~$n$-dimensional vector,~$\mt_0^i\in\mbb{R}^n$. For any arbitrary vector,~$\mt_0^i\in\mbb{R}^n$, we use~$\oplus \mt_0^i$ to denote the subspace spanned by~$\mt_0^i$, i.e. the collection of all~$\alpha\mt_0^i$, with~$\alpha\in\mbb{R}$. Similarly, for a matrix,~$A\in\mbb{R}^{n\times n}$, we use~$\oplus A$ to denote the (range space) subspace spanned by the columns of~$A$:
\begin{eqnarray*}
\oplus A = \left\{\sum_{i=1}^n\alpha_i \mb{a}_i~|~\alpha_i\in\mbb{R}\right\},\qquad
A = \left[
\begin{array}{ccc}
\mb{a}_1 & \ldots & \mb{a}_n
\end{array}
\right].
\end{eqnarray*}
For a collection of matrices,~$A_j\in\mbb{R}^{n\times n}$,~$j=1,\ldots,N$, we use~$\oplus_jA_j$ to denote the subspace spanned by all of the columns in all of the~$A_j$'s: let $A_j = \left[
\begin{array}{ccc}
\mb{a}_{j1} & \ldots & \mb{a}_{jn}
\end{array}
\right]$, then
\begin{eqnarray*}
\oplus_{j} A_j = \left\{\sum_{j=1}^N\beta_j\sum_{i=1}^n\alpha_i \mb{a}_{ji}~|~\alpha_i,\beta_j\in\mbb{R}\right\}.
\end{eqnarray*}

Let~$\rank(A)=\ul{\gamma},$ for some non-negative integer,~$\ul{\gamma}\leq n$, then~$\dim(\oplus A)=\rank(A)=\ul{\gamma}$. The pseudo-inverse of~$A$ is denoted by~$A^\dagger\in\mbb{R}^{n\times n}$; the orthogonal projection,~$\wt{\mt}_0^i$, of an arbitrary vector,~$\mt_0^i\in\mbb{R}^n$, on the range space,~$\oplus A$, is given by the matrix~$I_{A}=AA^\dagger$, i.e.
\begin{eqnarray}
\wt{\mt}_0^i = I_A\mt_0^i = AA^\dagger\mt_0^i.
\end{eqnarray}
With this notation,~$\wt{\mt}_0^i\in{\oplus A}\subseteq\mbb{R}^n$. Clearly,~$I_A ^2 = AA^\dagger AA^\dagger = AA^\dagger = I_A$ is a projection matrix from the properties of pseudo-inverse:~$AA^\dagger A = A$ and~$A^\dagger A A^\dagger = A^\dagger$. Note that when~$\mt_0^i\in{\oplus A}$, then~$I_A \mt_0^i=\mt_0^i.$

The Singular Value Decomposition (SVD) of~$A$ is given by~$A=U_AS_AV^\top_A$ with~$U_AU_A^\top=I_n,V_AV_A^\top=I_n$, then~$A^\dagger = VS_A^\dagger U^\top,$ where~$S_A^\dagger$ is the pseudo-inverse of the diagonal matrix of the singular values,~$S_A$ (with~$0^\dagger = 0$). When~$A$ is full-rank, we have~$A^\dagger=A^{-1},I_A =I_n$. Since~$\ol{\gamma} = \rank(A)$, the singular vectors ($U_A,V_A$) can be arranged such that
\begin{eqnarray}
I_A  & = AA^\dagger = U_AS_AV^\top_A V_AS_A^\dagger U_A^\top = U_AS_A S_A^\dagger U_A^\top,\\
& = U_A\left[
\begin{array}{cc}
\mb{0}_{\ol{\gamma}\times\ol{\gamma}}\\
&I_{\ul{\gamma}}
\end{array}
\right]
U_A^\top.
\end{eqnarray}
From the above, the projection matrix,~$I_A$, is symmetric with orthogonal eigenvectors (or left and right singular vectors),~$U_A$, such that its eigenvalues (singular values) are either~$0$'s or~$1$'s.

For some~$W=\{w_{ij}\}\in\mbb{R}^{N\times N}$ and some~$A=\{a_{ij}\}\in\mbb{R}^{n\times n}$ with~$w_{ij},a_{ij}\in\mbb{R}$, the matrix Kronecker product is
\begin{eqnarray}
W\otimes A =
\left[\begin{array}{cccc}
w_{11}A & w_{12}A & \ldots&w_{1N}A\\
\vdots & \vdots & \ddots&\vdots\\
w_{N1}A & w_{N2}A & \ldots&w_{NN}A\\
\end{array}\right],
\end{eqnarray}
which lies in~$\mbb{R}^{nN\times nN}$. It can be verified that~$I_N\otimes A$ is a block-diagonal matrix where each diagonal block is~$A$ with a total of~$N$ blocks. We have $W\otimes A = (W\otimes I_n) (I_N\otimes A).$ The following properties are useful in the context of this paper.
\begin{eqnarray}
\left(W\otimes I_n\right)\left(I_N\otimes A\right) &=& \left(I_N\otimes A\right)\left(W\otimes I_n\right),\\
\left(W\otimes I_n\right)^k &=& (W^k\otimes I_n),
\end{eqnarray}
for some non-negative integer,~$k$. More details on these notions can be found in~\cite{hornJ}.

\section{Problem Formulation}\label{pf}
We consider average consensus in a multi-agent network when the inter-agent communication is subject to unwanted interference, i.e. the desired communication,~$\mb{x}^j\in\mbb{R}^n$, from agent~$j\in\mc{V}$ to agent~$i\in\mc{V}$ has an additive term,~$\mb{z}^{ij}\in\mbb{R}^n$, resulting into agent $i$ receiving~$\mt^j + \mb{z}^{ij}$ from agent $j$. We consider the case when this unwanted interference is linear. In particular, every link,~$j\ra i$ or~$(i,j)\in\mc{E}$, incurs the following additive interference:
\begin{eqnarray}
\mb{z}^{ij} = \sum_{m\in\mc{V}}a_{ij}^m\Gamma_{ij}^m\mt^{m},
\end{eqnarray}
where:~$a_{ij}^m=1,$ if agent~$m\in\mc{V}$ interferes with~$j\ra i$, and~$0$ otherwise; and~$\Gamma_{ij}^m\in\mbb{R}^{n\times n}$ is the interference gain when~$m\in\mc{V}$ interferes with the~$j\ra i$ communication. What agent~$i$ actually receives from agent~$j$ is thus:
\begin{eqnarray}
\mt_k^j + \sum_{m\in\mc{V}} a_{ijk}^m \Gamma_{ijk}^m \mt_k^m,
\end{eqnarray}
at time~$k$, where the subscript `${ijk}$' introduces the time dependency on the corresponding variables, see Fig.~\ref{fig_gl}.
\begin{figure}[h!]
\centering
\includegraphics[width=2.5in]{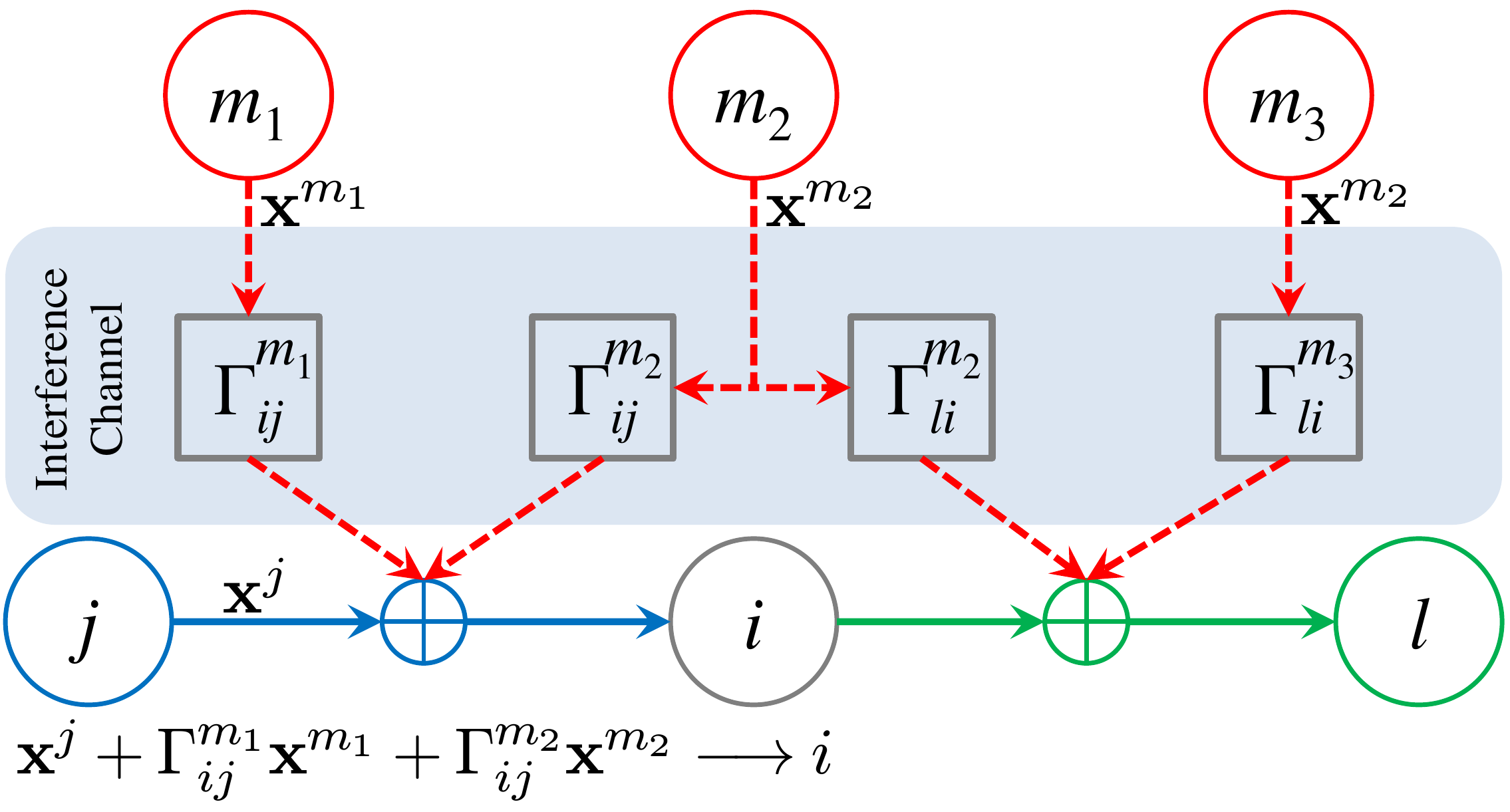}
\caption{Interference model: Note that agent~$j$ may also interfere with~$j\ra i$ communication, i.e.~$m_1$ or~$m_2$ can be~$j$. This may happen when agent~$j$'s transmission to agents other than agent~$i$ interfere with the~$j\ra i$ channel.}
\label{fig_gl}
\end{figure}

Given the interference setup, average-consensus implemented on the multi-agent network is given by
\begin{eqnarray}\label{cpi1}
\mb{x}_{k+1}^i = \sum_{j\in\mc{N}_i} w_{ij}\left(\mt_k^j + \sum_{m\in\mc{V}} a_{ijk}^m \Gamma_{ijk}^m \mt_k^m\right),
\end{eqnarray}
for~$k\geq0,i\in\mc{V}$, with~$\mt_0^i\in\mbb{R}^n$. Interference is only incurred when~$w_{ij}\neq0$, which is true for each~$j\in\mc{N}_i$, in general. In other words, interference is incurred on all the links that are allowed by the underlying communication graph,~$\mc{G}$. The protocol in Eq.~\eqref{cpi1} reduces to the standard average-consensus~\cite{boyd:04}, when there is no interference, i.e. when~${a}_{ijk}^m=0,$ for all~$i,j,k,m$, and converges to\footnote{See~\cite{boyd:04} for relevant conditions for convergence:~$W\mb{1}_n=\mb{1}_n, \mb{1}_n^\top W = \mb{1}_n^\top,~\mc{G}$ is strongly-connected, and~$w_{ij}\neq0$ for each~$(i,j)\in\mc{E}$.}
\begin{eqnarray}\label{pavg}
\mt_\infty^i \triangleq \lim_{k\ra\infty}\mt_k^i = \frac{1}{N}\sum_{j=1}^N\mt_0^j.
\end{eqnarray}
However, when there is interference, i.e.~$a_{ijk}^m\neq{0}$, Eq.~\eqref{cpi1}, in general, either goes to zero or diverges at all agents. The former is applicable when the effect of the interference results into a stable weight matrix,~$W=\{w_{ij}\}$, and the latter is in effect when the interference forces the spectral radius of the weight matrix to be greater than unity. The primary reason is that if~$w_{ij}$'s are chosen to sum to~$1$ in each neighborhood (to ensure~$W\mb{1}^\top=\mb{1}^\top$), their effective contribution in Eq.~\eqref{cpi2} is not~$1$ because of the unwanted interference. 

This paper studies appropriate modifications to Eq.~\eqref{cpi1} in order to achieve average-consensus. The design in this paper is based on a novel \emph{information alignment} principle that ensures that the spectral radius of the mixing matrix,~$W$, is not displaced form unity. We assume the following:
\begin{enumerate}[(a)]
\item \emph{No agent,~$i\in\mc{V}$, knows which (or how many) agents are interfering with its incoming or outgoing communication.}
\item \emph{The interference structure,~$a_{ijk}^m$ and~$\Gamma_{ijk}^m$, are constant over time,~$k$.} 

This assumption is to keep the exposition simple and is made without loss of generality as we will elaborate later.
\end{enumerate}

Under these assumptions, the \emph{standard average-consensus protocol} is given by
\begin{eqnarray}\label{cpi2}
\mb{x}_{k+1}^i = \sum_{j\in\mc{N}_i} w_{ij}\mt_k^j + \sum_{j\in\mc{N}_i} w_{ij}\sum_{m\in\mc{V}} a_{ij}^m\Gamma_{ij}^m \mt_k^m,
\end{eqnarray}
for~$k\geq0,\mt_0^i\in\mbb{R}^n$. The goal of this paper is to consider \emph{distributed averaging operations} in the presence of interference not only to establish the convergence, but further to ensure that the convergence is towards a meaningful quantity. To these aims, we present a conservative solution to this problem in Section~\ref{aca}, which is further improved in Sections~\ref{s_uoi} and~\ref{s_uii} for some practically relevant scenarios.

\section{A Conservative Approach}\label{aca}
Before considering the general case within a conservative paradigm, we explore a special case of uniform interference in Sections~\ref{s_ui} and~\ref{ill_ui}. We then provide the generalization in Section~\ref{ui_gen} and shed light on the conservative solution.

\subsection{Uniform Interference}\label{s_ui}
Uniform interference is when each communication link in the network experiences the same interference gain, i.e.~$\Gamma_{ij}^m = \Gamma_1,\forall i,j,m$. In other words, all of the blocks in the interference channel of Fig.~\ref{fig_gl} represent the same interference gain matrix,~$\Gamma_1\in\mbb{R}^{n\times n}$. In this context, Eq.~\eqref{cpi2} is given by 
\begin{eqnarray}\label{cpi3}
\mb{x}_{k+1}^i &=& \sum_{j\in\mc{N}_i} w_{ij}\mt_k^j + \sum_{m\in\mc{V}} b_i^m\Gamma_1 \mt_k^m,
\end{eqnarray}
where~$b_i^m = \sum_{j\in\mc{N}_i} w_{ij}a_{ij}^m$. Here,~$b_i^m\neq 0$ means that agent~$m\in\mc{V}$ interferes with agent~$i\in\mc{V}$ over some of the messages (from~$j\in\mc{N}_i$) received by agent~$i$. In fact, an agent~$m\in\mc{V}$ may interfere with agent~$i$'s reception on multiple incoming links, while an interferer,~$m$, may also belong to~$\mc{N}_i$, i.e. the neighbors of agent~$i$. To proceed with the analysis, we first write Eq.~\eqref{cpi2} in its matrix form: Let~$B_1$ be an~${N\times N}$ matrix whose~`$im$'th element is given by~$b_i^m$. Define the network state at time~$k$:
\begin{eqnarray}
\mt_{k} = \left[
\begin{array}{cccc}
\left(\mt_k^1\right)^\top&\left(\mt_k^2\right)^\top&\ldots&\left(\mt_k^N\right)^\top
\end{array}
\right]^\top.
\end{eqnarray}
Then, it can be verified that Eq.~\eqref{cpi3} is compactly written as 
\begin{eqnarray}
\mt_{k+1} 
&=& \left(W\otimes I_n+B_1\otimes \Gamma_1\right)\mt_k. \label{cpm_s1}
\end{eqnarray}
The~$N\times N$ weight matrix,~$W$, has the sparsity pattern of the consensus graph,~$\mc{G}$, while the~$N\times N$ matrix,~$B_1$, has the sparsity pattern of what can be referred to as the \emph{interference graph}--induced by the interferers. We have the following result.
\begin{lem}\label{lem1}
If~$\Gamma_1\mt_0^i = \mb{0}_n,\forall i$, then~$\Gamma_1\mt_k^i = \mb{0}_n,\forall i,k.$
\end{lem}
\begin{proof}
Note that~$\Gamma_1\mt_k^i$ is a local operation at the~$i$th agent. This is equivalent to multiplying~$I_N\otimes\Gamma_1$ with the network vector,~$\mt_k$. From the lemma's statement, we have~$(I_N\otimes \Gamma_1)\mt_0=\mb{0}_{nN}.$ Now note that (recall Section~\ref{pre_not})
\begin{eqnarray*}
\left(I_N\otimes\Gamma_1\right)\left(W\otimes I_n + B_1\otimes \Gamma_1\right) =\left(W\otimes \Gamma_1 + B_1\otimes \Gamma_1^2\right),\\
=\left(W\otimes I_n + B_1\otimes \Gamma_1\right) \left(I_N\otimes\Gamma_1\right).
\end{eqnarray*}
Subsequently, multiply both sides of Eq.~\eqref{cpm_s1} by~$\left(I_N\otimes\Gamma_1\right)$:
\begin{eqnarray*}
\left(I_N\otimes\Gamma_1\right)\mt_{k+1} = \left(W\otimes I_n + B_1\otimes \Gamma_1\right) \left(I_N\otimes\Gamma_1\right)\mt_k,\\
= \left(W\otimes I_n + B_1\otimes \Gamma_1\right)^{k+1} \left(I_N\otimes\Gamma_1\right)\mt_0=\mb{0}_n,
\end{eqnarray*}
and the lemma follows.
\end{proof}
The above lemma shows that the effect of \emph{uniform interference} can be removed from the average-consensus protocol if the data (initial conditions) lies in the null space of the interference,~$\Gamma_1$. To proceed, let us denote the interference null space (of $\Gamma_1$) by $\Theta_{\Gamma_1}$. Recall that~$\oplus_i \mt_0^i$ denotes the subspace spanned by all of the initial conditions, the applicability of Lemma~\ref{lem1} is not straightforward because:
\begin{inparaenum}[(i)]
\item~$\dim(\oplus_i \mt_0^i)>\dim(\Theta_{\Gamma_1})$, in general; and,
\item even when~$\dim(\oplus_i \mt_0^i)\leq\dim(\Theta_{\Gamma_1})$, the data subspace,~$\oplus_i \mt_0^i$, may not belong to the null space of the interference,~$\Theta_{\Gamma_1}$.
\end{inparaenum}
However, intuitively, a scheme can be conceived as follows: \emph{Project} the data on a low-dimensional subspace,~$\mc{S}$, such that~$\dim(\mc{S})\leq\dim(\Theta_{\Gamma_1})$; and, \emph{Align} this projected subspace,~$\mc{S}$, on the null-space,~$\Theta_{\Gamma_1}$, of the interference. At this point, we must ensure that this alignment is reversible so that its effect can be undone in order to recover the projected data subspace,~$\mc{S}$. To this aim, we provide the following lemma.

\begin{lem}\label{Tlem}
For some~$0\leq\ul{\gamma}\leq n$, let~$\Gamma_1\in\mbb{R}^{n\times n}$ have rank~$\ol{\gamma}=n-\ul{\gamma}$, and let another matrix,~$I_{\mc{S}}\in\mbb{R}^{n\times n}$ have rank~$\ul{\gamma}$. There exists a full-rank preconditioning,~$T_1\in\mbb{R}^{n\times n}$, such that~$\Gamma_1 T_1I_{\mc{S}}=\mb{0}_{n\times n}$.
\end{lem}

\begin{proof}
Since~$\Gamma_1$ has rank~$\ol{\gamma}$, there exists a singular value decomposition,~$\Gamma_1=U_1S_1V_1^\top$, where the~$n\times n$ diagonal matrix~$S_1$ is such that its first~$\ol{\gamma}$ elements are the singular values of~$\Gamma_1$, and the remaining~$\ul{\gamma}$ elements are zeros. With this structure on~$\mc{S}$, the matrix~$V_1$ can be partitioned into 
\begin{eqnarray}
V_1 = \left[
\begin{array}{ccc}
\ol{V}_1 & \ul{V}_1
\end{array}
\right],
\end{eqnarray}
(with~$\ol{V}_1\in\mbb{R}^{n\times \ol{\gamma}}$ and~$\ul{V}_1\in\mbb{R}^{n\times \ul{\gamma}}$), where~$\oplus\ul{V}_1$ is the null-space of~$\Gamma_1$. Similarly,~$I_{\mc{S}}=U_{\mc{S}}S_{\mc{S}}V_{\mc{S}}^\top$ with rank~$\ul{\gamma}$, where the matrices,~$U_{\mc{S}}$ and~$V_{\mc{S}}$, are arranged such that the first~$\ol{\gamma}$ diagonals of~$S_{\mc{S}}$ are zeros and the remaining are the~$\ul{\gamma}$ singular values of~$I_{\mc{S}}$. Define
\begin{eqnarray}\label{T0map}
T_1 = \left[
\begin{array}{ccc}
\ol{V}_1^\prime & \ul{V}_1^\prime
\end{array}
\right]U_{\mc{S}}^\top,
\end{eqnarray}
where~$\ul{V}_1^\prime$ is such that ~$\oplus\ul{V}_1^\prime = \oplus\ul{V}_1$, and~$\ol{V}_1^\prime$ is chosen arbitrarily such that~$T_1$ is invertile. With this construction, note that~$\ol{V}_1^\top\ul{V}_1^\prime$ is a zero matrix because~$\ol{V}_1$ is orthogonal to the column-span of~$\ul{V}_1$ (by the definition of the SVD). We have
\begin{eqnarray}\nonumber
\Gamma_1 T_1I_{\mc{S}} 
= U_1S_1\left[
\begin{array}{cc}
\ol{V}_1^\top\ol{V}_1^\prime & \mb{0}_{\ol{\gamma}\times\ul{\gamma}}\\
 \ul{V}_1^\top\ol{V}_1^\prime & \ul{V}_1^\top\ul{V}_1^\prime
\end{array}
\right] 
S_{\mc{S}}V_{\mc{S}}^\top
= U_1\mb{0}_{n\times n}V_{\mc{S}}^\top,
\end{eqnarray}
and the lemma follows. 
\end{proof}
The above lemma shows that the computation of the preconditioning only requires the knowledge of the (uniform) interference null-space,~$\Theta_{\Gamma_1}\triangleq\oplus{\ul{V}_1}$. Clearly,~$T_1=V_1U_{\mc{S}}^\top$ is a valid preconditioning as with this~$\Gamma_1T_1I_{\mc{S}}$ is a zero matrix, but this choice is more restrictive and not necessary. 

\emph{Information alignment}: Lemma~\ref{Tlem} further sheds light on the notion of \emph{information alignment}, i.e. the desired information sent by the transmitter can be projected and aligned in such a way that it is not distorted by the interference. Not only that the information remains unharmed, it can be recovered at the receiver as the preconditioning~$T$, is invertible. The following theorem precisely establishes the notion of information alignment with the help of Lemmas~\ref{lem1} and~\ref{Tlem}.
\begin{thm}[Uniform Interference]\label{cui_th}
Let~${\Theta_{\Gamma_1}}$ denote the null space of~$\Gamma_1$ and let~$\ul{\gamma}=\dim({\Theta_{\Gamma_1}})$. In the presence of uniform interference, the protocol in Eq.~\eqref{cpm_s1} recovers the average in a~$\ul{\gamma}$-dimensional subspace,~${\mc{S}}$, of~$\mbb{R}^n$, via an information alignment procedure based on the preconditioning. 
\end{thm}
\begin{proof}
Without loss of generality, we assume that~${\mc{S}}={\oplus A}$, where~${\oplus A}$ denotes the range space (column span) of some matrix,~$A\in\mbb{R}^{n\times n}$, such that~$\dim({\oplus A})=\ul{\gamma}$. Define $I_{\mc{S}} = A^\dagger A$, 
where~$I_{\mc{S}}$ is the orthogonal projection that projects any arbitrary vector in~$\mbb{R}^n$ on~${\mc{S}}$. Define the projected (on~${\mc{S}}$) and transformed initial conditions:~$\wh{\mt}_0^i\triangleq T_1I_{\mc{S}}\mt_0^i,\forall i\in\mc{V}$, where~$T_1$ is the invertible preconditioning given in Lemma~\ref{Tlem}. From Lemma~\ref{Tlem}, we have
\begin{eqnarray}
\Gamma_1\wh{\mt}_0^i=\Gamma_1T_1I_{\mc{S}}\mt_0^i=\mb{0}_n,\qquad\forall i\in\mc{V},
\end{eqnarray}
i.e. the alignment makes the initial conditions invisible to the interference. From Lemma~\ref{lem1}, Eq.~\eqref{cpm_s1} reduces to $
\wh{\mt}_{k+1}^i = \sum_{j\in\mc{N}_i}w_{ij}\wh{\mt}^j_k$, when the initial conditions are~$\wh{\mt}_0^i, \forall i\in\mc{V}$, which converges to the average of the \emph{transformed and projected} initial conditions,~$\wh{\mt}_0^i$'s, under the standard average-consensus conditions on~$\mc{G}$ and~$W$. Finally, average in~${\mc{S}}$ is recovered by
\begin{eqnarray*}
\wt{\mt}_\infty^i = T_1^{-1}\wh{\mt}_{\infty}^i = \frac{1}{N}\sum_{j=1}^NT_1^{-1}\wh{\mt}_0^j = \frac{1}{N}\sum_{j=1}^N I_{\mc{S}}\mt_0^j,\qquad \forall i\in\mc{V},
\end{eqnarray*}
and the theorem follows.
\end{proof}
The above theorem shows that in the presence of uniform interference, a careful information alignment results into obtaining the data (initial conditions) average projected onto any arbitrary~$\ul{\gamma}$-dimensional subspace,~$\mc{S}$, of~$\mbb{R}^n$. We note that a completely distributed application of Theorem~\ref{cui_th} requires that each agent knows the null-space,~$\Theta_{\Gamma_1}$, of the (uniform) interference, recall Lemma~\ref{Tlem}; and thus is completely local. In addition, all of the agents are required to agree on the desired signal subspace,~$\mc{S}$, where the data is to be projected. 

\subsection{Illustration of Theorem~\ref{cui_th}}\label{ill_ui}
In essence, Theorem~\ref{cui_th} can be summarized in the following steps, illustrated with the help of Fig.~\ref{f11}:
\begin{enumerate}[(i)]
\item \emph{Project} the data,~$\mbb{R}^n$, on a~$\ul{\gamma}$-dimensional subspace,~${\mc{S}}$, via the projection matrix,~$I_{\mc{S}}$.
    
    In Fig.~\ref{f11} (a), the data (initial conditions) lies arbitrarily in~$\mbb{R}^3$ projected on a~$\ul{\gamma}=2$-dimensional subspace,~${\mc{S}}$, in Fig.~\ref{f11} (b). Interference is given by a rank~$1$ matrix,~$\Gamma_1$; the interference subspace is shown by the black line;
\item \emph{Align} the projected subspace,~${\mc{S}}$, on the null space,~${\Theta_{\Gamma_1}}$, of interference,~$\Gamma_1$, via the preconditioning,~$T_1$.
    
    In Fig.~\ref{f11} (c), the projected subspace,~${\mc{S}}$, is aligned to the null of space,~${\Theta_{\Gamma_1}}$, of the interference via preconditioning with~$T_1$. Note that after the alignment, the data is orthogonal to the interference subspace (black line);
\item \emph{Consensus} is implemented now on the null space of the interference, see Fig.~\ref{f11} (d).
\item \emph{Recover} the average in~${\mc{S}}$ via~$T_1^{-1}$.

Finally, the average in the null space,~${\Theta_{\Gamma_1}}$, is translated back to the the signal subspace,~${\mc{S}}$, via~$T_1^{-1}$. We also show the true average in~$\mbb{R}^3$ by the `$\star$', see Fig.~\ref{f11} (e).
\end{enumerate}
\begin{figure*}
\centering
\subfigure{\includegraphics[width=1.25in]{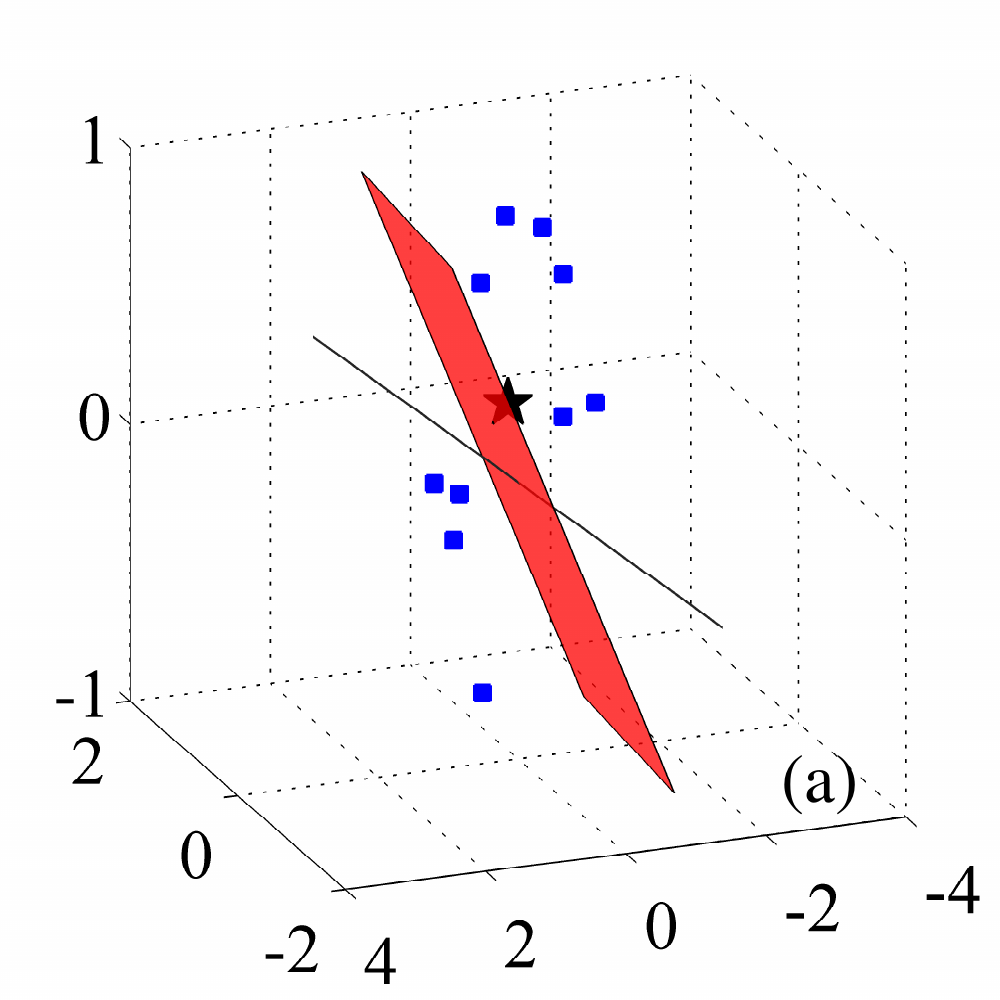}}
\subfigure{\includegraphics[width=1.25in]{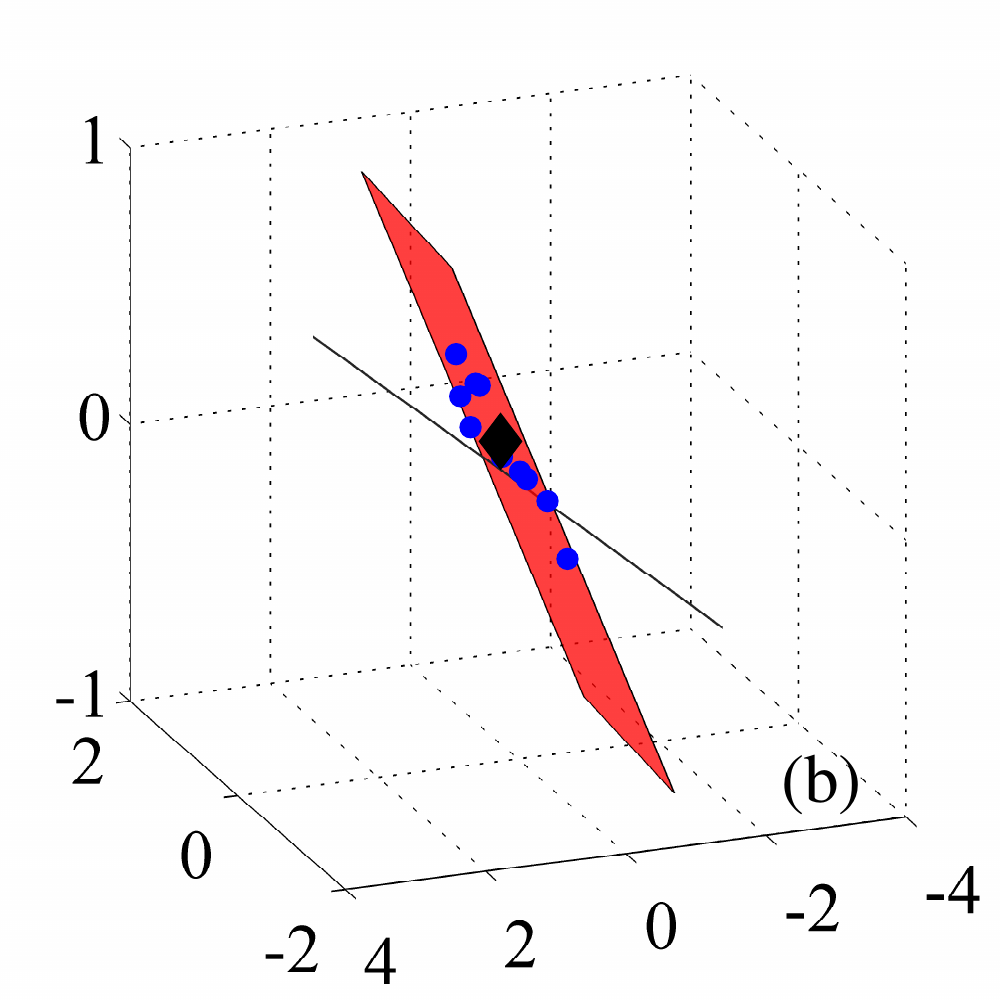}}
\subfigure{\includegraphics[width=1.25in]{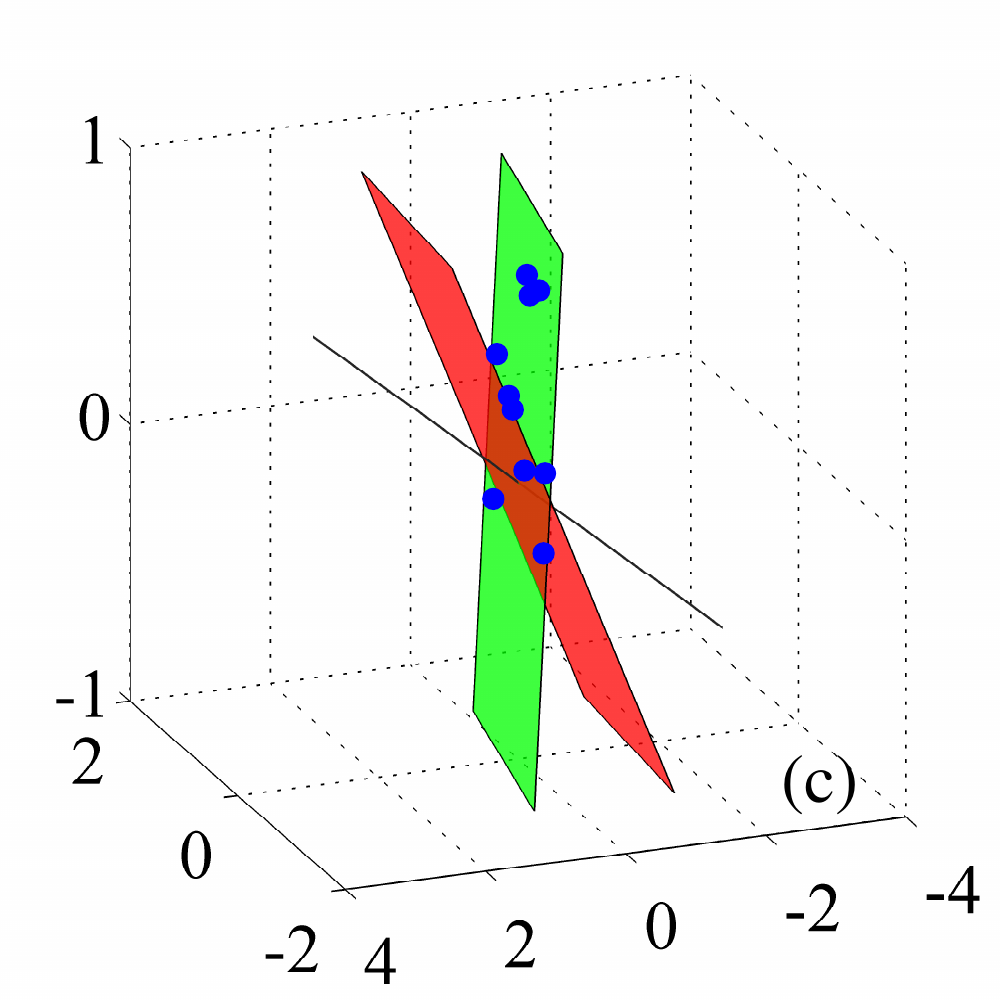}}
\subfigure{\includegraphics[width=1.25in]{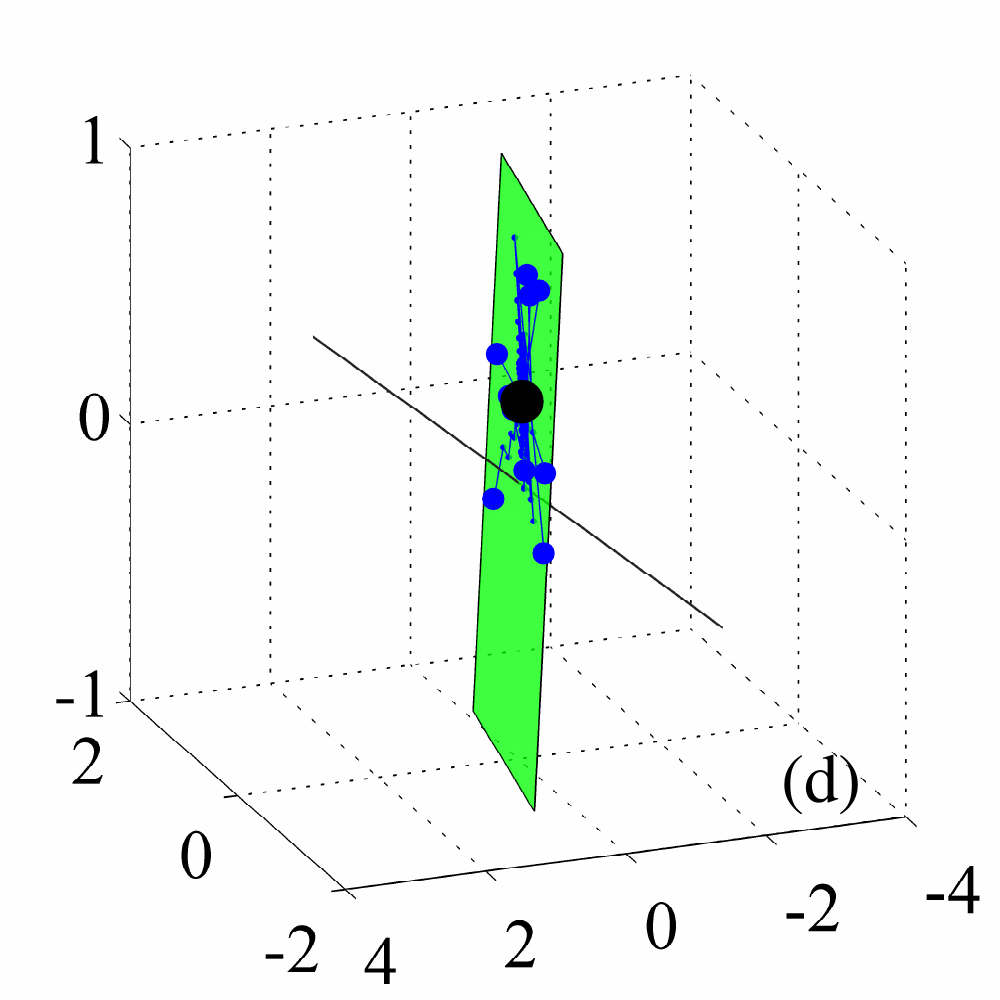}}
\subfigure{\includegraphics[width=1.25in]{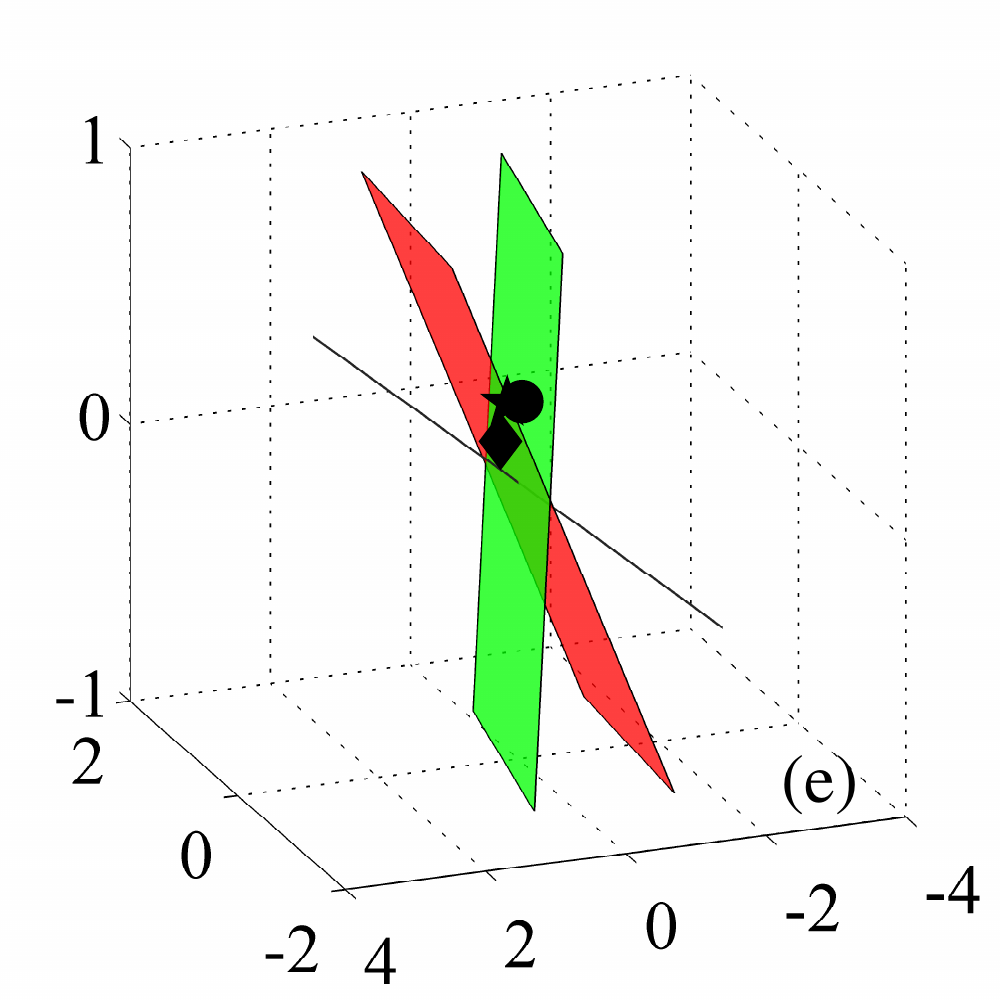}}
\caption{Consensus under uniform interference: (a) Signal space,~$\mbb{R}^3$, data shown as squares and the average as `$\star$'; (b) Projected signal subspace,~${\mc{S}}$, shown as circles and the average as `$\diamond$'; (c) Alignment on the null space of the interference,~$T_1I_{{\mc{S}}}\mt_0^i$; (d) Consensus in the null space of the interference,~$\wh{\mt}_k^i$, average shown as large filled circle; and, (e) Translation back to the signal subspace,~$T_1^{-1}\wh{\mt}_{\infty}^i$.}
\label{f11}
\end{figure*}

From Theorem~\ref{cui_th}, when~$\Gamma_1$ is full-rank, i.e.~$\ul{\gamma}=0,$ the iterations converge to a zero-dimensional subspace and are not meaningful. However, if the interference is low-rank, consensus under uniform interference may still remain meaningful. In fact, we can establish the following immediate corollaries.
\begin{cor}[Perfect Consensus]\label{cor1}
Let~$\mt_0^i\in\mbb{R}^n$ be such that~$\dim(\oplus_i \mt_0^i)\leq\dim(\Theta_{\Gamma_1})$. Then consensus under uniform interference, Eq.~\eqref{cpm_s1}, recovers the true average of the initial conditions,~$\mt_0^i$.
\end{cor}
\begin{cor}[Principal/Selective Consensus]
Let the initial conditions,~$\mt_0^i$, belong to the range space,~$\oplus A$, of some matrix,~$A\in\mbb{R}^{n\times n}$. Then consensus under uniform interference, Eq.~\eqref{cpm_s1}, recovers the average in a~$\ul{\gamma}=\dim(\Theta_{\Gamma_1})$ subspace that can be chosen along any~$\ul{\gamma}$ singular values of~$A$.
\end{cor}
The proofs of the above two corollaries immediately follow from Theorem~\ref{cui_th}. In fact, the protocol, Eq.~\eqref{cpm_s1}, can be tailored towards the~$\ul{\gamma}$ largest singular values (principal consensus), or towards any arbitrary~$\ul{\gamma}$ singular values (selective consensus). The former is applicable to the cases when the data (initial conditions) lies primarily along a few singular values. While the latter is applicable to the cases when the initial conditions are known to have meaningful components in some singular values. We now show a few examples on this approach.

\begin{ex}
Consider the initial conditions,~$\mt_0^i,\forall i$, to lie in the range space,~$\oplus A$, with the following:
\begin{eqnarray}\label{eq_ill1}
A = \left[
\begin{array}{cc}
1 & 1\\
1 & 1
\end{array}
\right],
I_{\mc{S}}= \frac{1}{2}\left[
\begin{array}{cc}
\frac{1}{2}&\frac{1}{2}\\
\frac{1}{2}&\frac{1}{2}
\end{array}
\right], U_{\mc{S}}= \left[
\begin{array}{rr}
\frac{-1}{\sqrt{2}}&\frac{-1}{\sqrt{2}}\\
\frac{1}{\sqrt{2}}&\frac{-1}{\sqrt{2}}
\end{array}
\right].
\end{eqnarray}
Clearly,~$\dim({\oplus A})=1$. Consider any rank~$1$ interference,~$\Gamma$:
\begin{eqnarray*}
\Gamma_1 = \alpha\left[
\begin{array}{cc}
1&1\\
1&1
\end{array}
\right],{\Theta_{\Gamma_1}} = \beta
\left[
\begin{array}{rr}
1\\
-1
\end{array}
\right],\qquad\alpha,\beta\in\mbb{R}.
\end{eqnarray*}
It can be easily verified that originally the data subspace,~$\oplus A$, is aligned with the interference subspace,~$\oplus \Gamma_1$, and standard consensus operation is not applicable as no agent knows from which agents and on what links this interference is being incurred (recall Assumption (a) in Section~\ref{pf}). In other words, each agent~$i$, implementing Eq.~\eqref{cpi1}, cannot ensure that~$\sum_{j\in\mc{N}_i}w_{ij} + \sum_{j\in\mc{N}_i}w_{ij}\sum_{m\in\mc{V}} a_{ij}^m =1$ for the above iterations to remain meaningful and convergent.

Following Theorem~\ref{cui_th}, we choose $T_1 = V_1U^\top_{{\mc{S}}}$, which can be verified to be a diagonal matrix with $1$ and $-1$ on the diagonal, resulting into~$\Gamma_1 T_1I_{\mc{S}}=\mb{0}_{2\times 2}$. The effect of preconditioning,~$T_1$, is to move the entire~$1$-dimensional signal subspace in the null space of the interference. Subsequently, 
\begin{eqnarray*}
\wh{\mt}_{k+1}^i = \sum_{j\in\mc{N}_i} w_{ij}\wh{\mt}^j_k + \sum_{m\in\mc{V}} b_{i}^{m}\Gamma_1\wh{\mt}^{m}_k = \sum_{j\in\mc{N}_i} w_{ij}\wh{\mt}^j_k + \mb{0}_n,
\end{eqnarray*}
when~$\wh{\mt}_0^i=T_1I_{\mc{S}}\mt_0^i=T_1\mt_0^i$, and true average is recovered via~$T_1^{-1}$ (see Corollary~\ref{cor1}).
\end{ex}
%

\subsection{A Conservative Generalization}\label{ui_gen}
In Section~\ref{s_ui}, we assume that the overall interference structure, recall Fig.~\ref{fig_gl}, is such that the interference gains are uniform, i.e.~$\Gamma_{ij}^m=\Gamma_1.$ We now provide a conservative generalization of Theorem~\ref{cui_th} to the case when the interferences do not have a uniform structure.
\begin{thm}\label{con_th}
Define~$\Gamma\in\mbb{R}^{n\times n}$ to be the network interference matrix such that
\begin{eqnarray}
\oplus_{i,j,m}{\Gamma_{ij}^m}~~\subseteq~~\oplus\Gamma,\qquad i,j,m,\in\mc{V}.
\end{eqnarray}
Let~${\Theta_{\Gamma}}$ be the null space of~$\Gamma$ with~$\ul{\gamma}=\dim({\Theta_{\Gamma}})$. The protocol in Eq.~\eqref{cpi2} recovers the average in a~$\ul{\gamma}$-dimensional subspace,~${\mc{S}}$, of~$\mbb{R}^n$, with an appropriate alignment.
\end{thm}
The proof follows directly from Lemmas~\ref{lem1},~\ref{Tlem}, and Theorem~\ref{cui_th}. Following the earlier discussion, we choose a global preconditioning,~$T\in\mbb{R}^{n\times n}$, based on the null-space,~$\Theta_\Gamma$, of the network interference,~$\Gamma$. The solution described by Theorem~\ref{con_th} requires each interference to belong to some subspace of the network interference,~${\oplus}\Gamma$, and each agent to have the knowledge of this network interference. However, this global knowledge is not why the approach in Theorem~\ref{con_th} is \emph{conservative}, as we explain below.

Consider~${\oplus}_{i,j,}{\Gamma_{ij}^m}\subseteq\mbb{R}^n$, to be such that~$\dim\left({\oplus}{\Gamma_{ij}^m}\right)=1$, for each~$i,j,m\in\mc{V}$. In other words, each interference block in Fig.~\ref{fig_gl} is a one-dimensional line in~$\mbb{R}^n$. Theorem~\ref{con_th} assumes a network interference matrix,~$\Gamma$, such that its range space,~${\oplus}\Gamma$, includes every local interference subspace,~${\oplus}{\Gamma_{ij}^m}$. When each local interference subspace,~${\oplus}{\Gamma_{ij}^m}$, is one-dimensional, we can easily have~$\dim({\oplus}_{i,j,m}\Gamma_{ij}^m)=n$, subsequently requiring~$\dim(\oplus\Gamma)=n$. This happens when the local interference subspaces are not aligned perfectly. Theorem~\ref{cui_th} is a very special scenario when all of the local interference subspaces are exactly the same (perfectly aligned). Extending it to Theorem~\ref{con_th}, however, shows that when the local interference are misaligned,~${\oplus}\Gamma$ may have dimension~$n$, and consensus is only ensured on a zero-dimensional subspace, i.e. with~$I_{\mc{S}} = \mb{0}_{n\times n}$.

This limitation of Theorem~\ref{con_th} invokes a significant question: \emph{When all of the local interferences are misaligned such that their collection spans the entire~$\mbb{R}^n$, can consensus recover anything meaningful?} Is it true that Theorem~\ref{con_th} is the only candidate solution? In the next sections, we show that there are indeed \emph{distributed and local} protocols that can recover meaningful information. To proceed, we add another assumption, (c), to Assumptions~(a) and (b) in Section~\ref{pf}:
\begin{enumerate}[(c)]
\item \emph{The interference matrices,~$\Gamma_{ij}^m$, are independent over~$j$}. 
    
    Note that in our interference model, any agent~$m\in\mc{V}$ can interfere with~$j\ra i$ communication; from Assumption (a), these are unknown to either agent~$j$ or~$i$. Assumption (c) is equivalent to saying that this interference is only a function of the interferer,~$m\in\mc{V}$, or the receiver,~$i\in\mc{V}$, and is independent of communication link,~$j\ra i$.
\end{enumerate}
We consider the design and analysis in the following cases:

\emph{Uniform Outgoing Interference}:~$\Gamma_{i}^m=\Gamma_m,\forall i,m\in{\mc{V}}$. In this case, each agent,~$m\in\mc{V}$, interferes with every other agent via the same interference matrix,~$\Gamma_m$, see Fig.~\ref{uoi_T} (top). This case is discussed in Section~\ref{s_uoi};

\emph{Uniform Incoming Interference}:~$\Gamma_{i}^m = \Gamma_i,\forall i,m\in{\mc{V}}.$ In this case, each agent~$i$ incurs the same interference,~$\Gamma_i$, over all the interferers,~$m\in\mc{V}$, see Fig.~\ref{uoi_T} (bottom). This case is discussed in Section~\ref{s_uii}.

\begin{figure}[h!]
\centering
\subfigure{\includegraphics[width=2.5in]{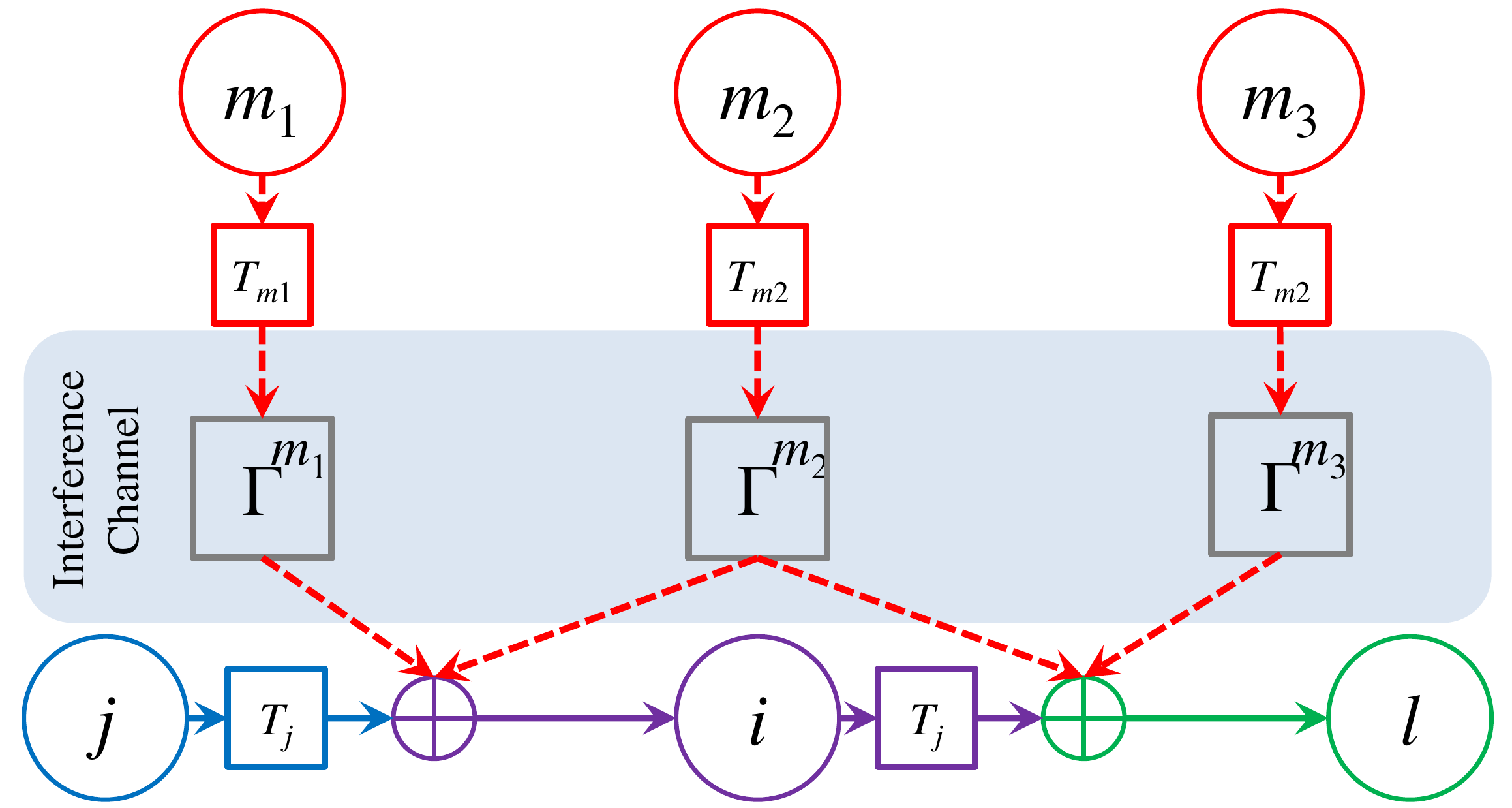}}
\hspace{1cm}
\subfigure{\includegraphics[width=2.5in]{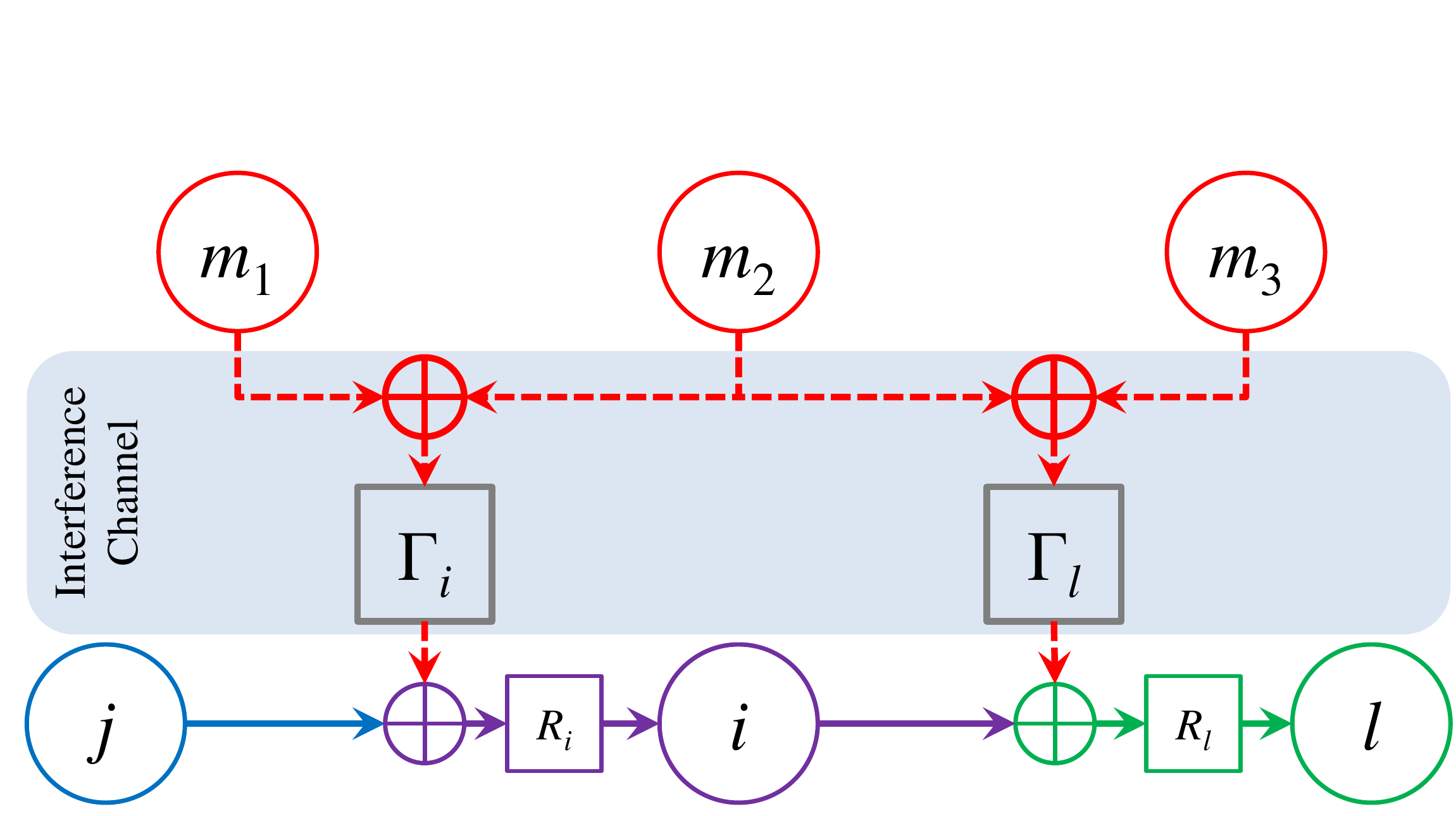}}
\caption{(Top) Uniform Outgoing (Bottom) Uniform Incoming. The blocks, $T_i$'s and $R_i$'s, will become clear from Sections~\ref{s_uoi} and \ref{s_uii}.}
\label{uoi_T}
\end{figure}

\section{Uniform Outgoing Interference}\label{s_uoi}
This section presents results for the uniform outgoing interference, i.e. each agent,~$m\in\mc{V}$, interferes with every other agent in the same way. Recall that agent~$j$ wishes to transmit~$\mt^j$ to agent~$i$ in the presence of interference. When this interference depends only on the interfere, agent~$i$ receives
\begin{eqnarray}
\mt_k^j + \sum_{m\in\mc{V}} a_{ij}^m\Gamma_m \mt_k^m,
\end{eqnarray}
from agent~$j$ at time~$k$. We modify the transmission as~$T_m\wt{\mt}_k^m,$ for all~$m\in\mc{V}$ for some auxiliary state variable,~$\wt{\mt}_k^i\in\mbb{R}^{n}$, to be explicitly defined shortly; agent~$i$ thus receives
\begin{eqnarray}
T_j\wt{\mt}_k^j + \sum_{m\in\mc{V}} a_{ij}^m\Gamma_m T_m\wt{\mt}_k^m,
\end{eqnarray}
from agent~$j$ at time~$k$. Consider the following protocol:
\begin{eqnarray}\label{cpi_uoia}
\wt{\mt}_{k+1}^i = \sum_{j\in\mc{N}_i} W_{ij}\left(T_j\wt{\mt}_k^j + \sum_{m\in\mc{V}} a_{ij}^m\Gamma_m T_m\wt{\mt}_k^m\right),
\end{eqnarray}
where~$W_{ij}\in\mbb{R}^{n\times n}$ is now a \emph{matrix} that agent~$i$ associates with agent~$j$; recall that earlier~$W_{ij} = w_{ij}I_n$. We get
\begin{eqnarray}\label{cpi_uoib}
\wt{\mt}_{k+1}^i = \sum_{j\in\mc{N}_i} W_{ij}T_j\wt{\mt}_k^j + \sum_{m\in\mc{V}}B_{im} \Gamma_m T_m\wt{\mt}_k^m,
\end{eqnarray}
where~$B_{im} = \sum_{j\in\mc{N}_i} W_{ij}a_{ij}^m$. We have the following result.
\begin{lem}\label{Tlem_uoi}
For some non-negative integer,~$\ul{\gamma}\leq n$, let each outgoing interference matrix,~$\Gamma_i$, have rank~$\ol{\gamma}\triangleq n-\ul{\gamma}$. Let~$I_{\mc{S}}\in\mbb{R}^{n\times n}$ be the projection matrix that projects~$\mbb{R}^n$ on~${\mc{S}}$, where~$\dim({\mc{S}})=\ul{\gamma}$. Then, there exist~$T_i$ at each~$i\in\mc{V}$, and~$W_{ij}$'s for all~$(i,j)\in\mc{E}$ such that Eq.~\eqref{cpi_uoib} becomes
\begin{eqnarray}\nonumber
\wt{\mt}_{k+1}^i = \sum_{j\in\mc{N}_i} w_{ij}\wt{\mt}_k^j,
\end{eqnarray}
at each~$i\in\mc{V}$, when~$\wt{\mt}_0^i\in {\mc{S}}$.
\end{lem}
\begin{proof}
Without loss of generality, we assume that~${\mc{S}}={\oplus A}$, where~${\oplus A}$ denotes the range space of some matrix,~$A\in\mbb{R}^{n\times n}$, such that~$\dim({\oplus A})=\ul{\gamma}$. Define $I_{\mc{S}} = A^\dagger A$, where~$I_{\mc{S}}$ is the orthogonal projection that projects any arbitrary vector in~$\mbb{R}^n$ on~${\mc{S}}$. Define~$\wt{\mt}_0^i$ to be the projected initial conditions, i.e.~$\wt{\mt}_0^i\triangleq I_{\mc{S}}\mt_0^i$. Let~$T_i$ be the \emph{locally designed}, invertible preconditioning, obtained at each~$i\in\mc{V}$ from the null-space,~$\Theta_{\Gamma_i}$, of its outgoing interference matrix,~$\Gamma_i$, see Lemma~\ref{Tlem}. Clearly, following Lemma~\ref{Tlem}, we have $\Gamma_iT_i\wt{\mt}_0^i=\mb{0}_n,\forall i\in\mc{V}$. Choose
\begin{eqnarray}
W_{ij}=w_{ij}T_j^{-1}.
\end{eqnarray}
From Eq.~\eqref{cpi_uoib}, we have
\begin{eqnarray}\nonumber
\wt{\mt}_{k+1}^i = \sum_{j\in\mc{N}_i} w_{ij}\wt{\mt}_k^j + \sum_{m\in\mc{V}} B_{im} \Gamma_m T_m\wt{\mt}_k^m.
\end{eqnarray}
We claim that when~$\wt{\mt}_0^i\in\mc{S},\forall i\in\mc{V}$, then~$\wt{\mt}_k^i\in {{\mc{S}}},\forall i\in\mc{V},k$, proven below by induction. Consider~$k=0$, then
\begin{eqnarray}\nonumber
\wt{\mt}_{1}^i = \sum_{j\in\mc{N}_i} w_{ij}\wt{\mt}_0^j + \sum_{m\in\mc{V}} B_{im} \Gamma_m T_m\wt{\mt}_0^m = \sum_{j\in\mc{N}_i} w_{ij}\wt{\mt}_0^j,
\end{eqnarray}
which is a linear combination of vectors in~${\mc{S}}$ and thus lies in~${\mc{S}}$. Assume that~$\wt{\mt}_k^i\in{\mc{S}},\forall i\in\mc{V}$, and some~$k$, leading to~$\Gamma_iT_i\wt{\mt}_k^i=\mb{0}_n.$ Then for~$k+1$:
\begin{eqnarray}\nonumber
\wt{\mt}_{k+1}^i = \sum_{j\in\mc{N}_i} w_{ij}\wt{\mt}_k^j + \sum_{m\in\mc{V}} B_{im} \Gamma_m T_m\wt{\mt}_k^m = \sum_{j\in\mc{N}_i} w_{ij}\wt{\mt}_k^j,
\end{eqnarray}
which is a linear combination of vectors in~${\mc{S}}$.
\end{proof}
\noindent The main result on uniform outgoing interference is as follows.
\begin{thm}\label{th_uoi}
Let~$\Theta_{\Gamma_i}$ denote the null space of~$\Gamma_i$, and let~$\ul{\gamma}\triangleq\min_{i\in\mc{V}}\{\dim(\Theta_{\Gamma_i})\}$. In the presence of uniform outgoing interference, Eq.~\eqref{cpi_uoia} recovers the average in a~$\ul{\gamma}$-dimensional subspace,~${\mc{S}}$, of~$\mbb{R}^n$, when we choose~$T_i$ according to Lemma~\ref{Tlem}, and~$W_{ij}=w_{ij}T_j^{-1}$, at each~$i,j\in\mc{N}_i$.
\end{thm}
The proof follows from Lemma~\ref{Tlem_uoi}. In other words, the consensus protocol in the presence of uniform outgoing interference, Eq.~\eqref{cpi_uoia}, converges to
\begin{eqnarray}
\wt{\mt}_\infty^i = \frac{1}{N}\sum_{j=1}^N \wt{\mt}^j_0 = \frac{1}{N}\sum_{j=1}^NI_{{\mc{S}}}\mt^j_0,
\end{eqnarray}
for any~$\mt_0^i\in\mbb{R}^n,\forall i\in\mc{V}.$ We note that each agent,~$i\in\mc{V}$, is only required to know the null-space of its outgoing interference,~$\Gamma_i$, to construct an appropriate preconditioning,~$T_i$. In addition, each agent,~$i\in\mc{V}$, is required to obtain the local pre-conditioners,~$T_j$'s, \emph{only} from its neighbors,~$j\in\mc{N}_i$; and thus, this step is also completely local. 

The protocol described in Theorem~\ref{th_uoi} can be cast in the purview of Fig.~\ref{uoi_T} (top). Notice that a transmission from any agent,~$i\in\mc{V}$, passes through agent~$i$'s dedicated preconditioning matrix,~$T_i$. The network (both non-interference and interference) sees only~$T_i\mt_k^i$ at each~$k$. Since the interference is a function of the transmitter (uniform outgoing), all of the agents ensure that a particular signal subspace,~$\mc{S}$, is not corrupted by the interference channel. The significance here is that even when the interferences are misaligned such that~${\oplus}_{i\in\mc{V}}{\Gamma_i} = \mbb{R}^n$, the protocol in Eq.~\eqref{cpi_uoia} recovers the average in~$\ul{\gamma}=\min_{i\in\mc{V}}\{\Theta_{\Gamma_i}\}$ dimensional signal subspace. On the other hand, the null space of the entire collection,~${\oplus}_{i\in\mc{V}}{\Gamma_i}$, may very well be~$0$-dimensional. For example, if each~$\Gamma_i$ is rank~$1$ such that each of the corresponding one-dimensional subspace is misaligned, Eq.~\eqref{cpi_uoia} recovers the average in an~$n-1$ dimensional signal subspace. On the other hand, Theorem~\ref{con_th} does not recover anything other than~$\mb{0}_n$.

\subsection{Illustration of Theorem~\ref{th_uoi}}
Let the initial conditions belong to a~$2$-dimensional subspace in~$\mbb{R}^3$ and consider~$N=10$ agents, with random initial conditions, shown as blue squares in Fig.~\ref{f11_uoi} (a). Uniform outgoing interference is chosen as one of the three~$1$-dimensional subspaces such that each interference appears at some agent in the network, see Fig.~\ref{f11_uoi} (b). Clearly, each interference is misaligned and~$\dim({\oplus}_i{\Gamma_i})=n=3.$ Hence, the protocol following Theorem~\ref{con_th} requires the signal subspace to be~$n-\dim({\oplus}_i{\Gamma_i})=0$ dimensional. However, when the agent transmissions are preconditioned using~$T_i$'s, each agent projects its transmission on the null space of its interference. Each receiver,~$i\in\mc{V}$, receives a misaligned data,~$T_j\mt^j$, from each of its neighbors,~$j\in\mc{N}_i$, see Fig.~\ref{f11_uoi} (c). Since each~$T_j\mt^j$ is a function of the corresponding neighbor,~$j$, the data can be translated back to~$\mc{S}$ via~$T_j^{-1}$, which is incorporated in the consensus weights,~$W_{ij}=w_{ij}T_j^{-1}$.
\begin{figure*}
\centering
\subfigure[]{\includegraphics[width=1.5in]{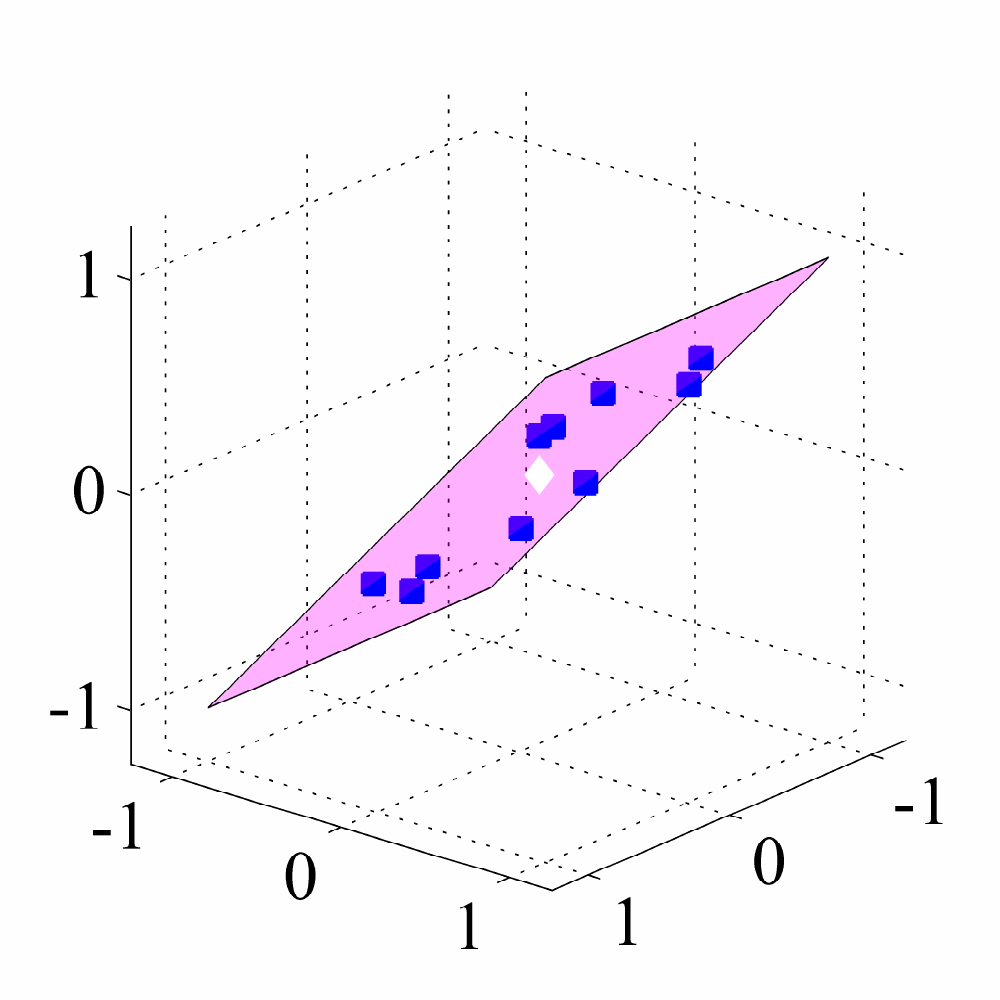}}
\subfigure[]{\includegraphics[width=1.5in]{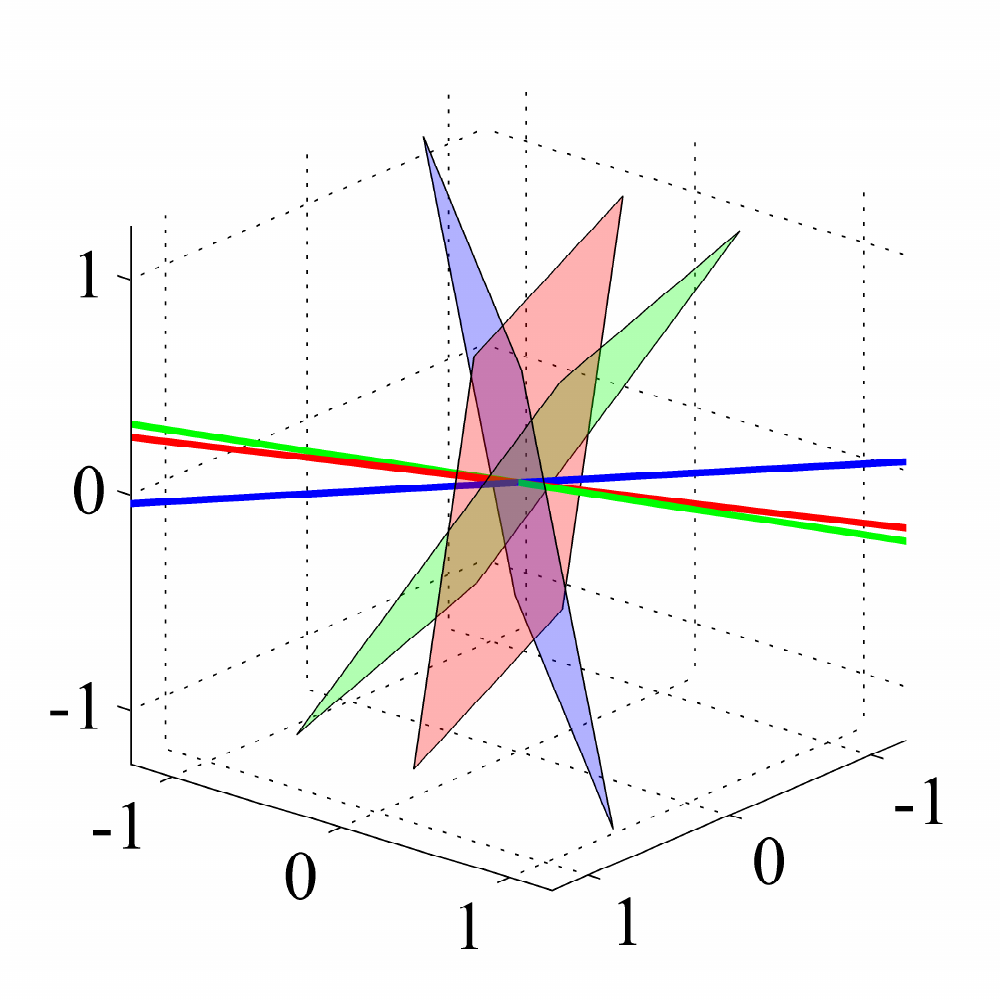}}
\subfigure[]{\includegraphics[width=1.5in]{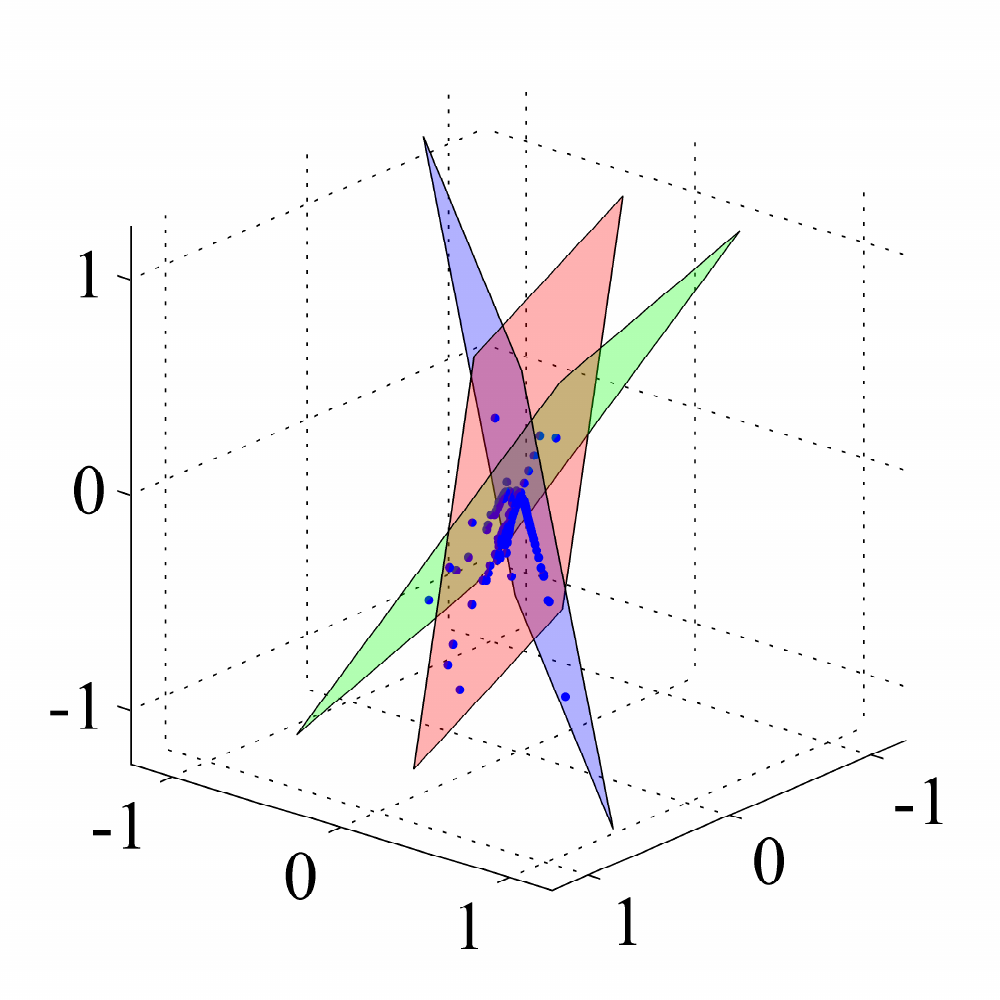}}
\subfigure[]{\includegraphics[width=1.5in]{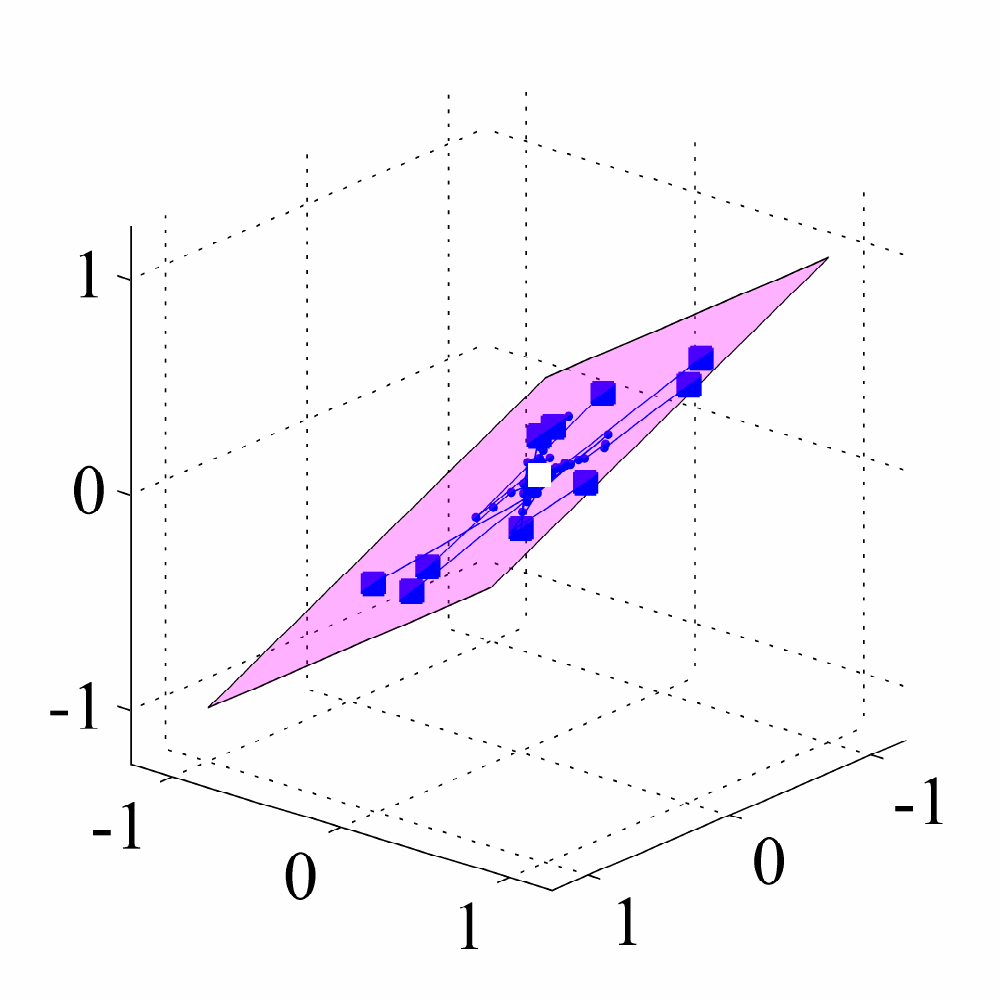}}
\caption{Consensus under uniform outgoing interference: (a) Signal space,~${\mc{S}} \subseteq \mbb{R}^3$, where~$\dim({\mc{S}})=2$; (b) One-dimensional range spaces,~${\oplus}{\Gamma_i}$, of~$\Gamma_i$'s--the null spaces of each are~$\ul{\gamma}=2$-dimensional, shown as planes; (c) Agent transmissions aligned in the corresponding null spaces over time,~$k$; (d) Consensus in the signal subspace,~${\mc{S}}$, after appropriate translations, at each~$i\in\mc{V}$, back to the signal subspace by~$T_j^{-1}$, with~$j\in\mc{N}_i$.}
\label{f11_uoi}
\end{figure*}

\section{Uniform Incoming Interference}\label{s_uii}
In this section, we consider the case of uniform incoming interference, i.e. each agent~$i\in\mc{V}$ incurs the same interference,~$\Gamma_i$, over all of the interferers,~$m\in\mc{V}$. This scenario is shown in Fig.~\ref{uoi_T} (bottom). We note that Theorem~\ref{con_th} is applicable here but results into a conservative approach as elaborated earlier. We note that this case is completely different from the uniform outgoing case (of the previous section), since preconditioning (alone) may not work as we explain below.

When an agent,~$m\in\mc{V}$, employs preconditioning, it may not precondition to account for the interference,~$\Gamma_i$, experienced at each receiver,~$i$, with which~$m$ may interfere. In the purview of Fig.~\ref{uoi_T} (bottom), if agent~$m_2\in\mc{V}$ preconditions using~$T_{m_2}$ to cancel the interference,~$\Gamma_i$, experienced by agent~$i$; the same preconditioning,~$T_{m_2}$, is not helpful to agent~$l$. For example, let agent~$m_2$ choose~$T_{m_2}=V_iU_{{\mc{S}}}^\top$ (a valid choice following Lemma~\ref{Tlem}), then as discussed earlier~$\Gamma_iV_iU_{{\mc{S}}}^\top I_{\mc{S}}=\mb{0}_{n\times n}$ and~$m_2$'s interference is not seen by agent~$i$. However, this preconditioning appears as~$\Gamma_lV_iU_{{\mc{S}}}^\top I_{\mc{S}}$ at agent~$l$, which is~$\mb{0}_{n\times n}$ only when~$V_l^\top V_i=I_n$. This is not true in general.

We now explicitly address the uniform incoming interference scenario. In this case, Eq.~\eqref{cpi2} takes the following form:
\begin{eqnarray}\label{cpi_uii_s}
\mb{x}_{k+1}^i = \sum_{j\in\mc{N}_i} W_{ij}\left(\mt_k^j + \Gamma_{i}\sum_{m\in\mc{V}} a_{ij}^m \mt_k^m\right), 
\end{eqnarray}
$ k\geq0,\mt_0^i\in\mbb{R}^n$ and where, as in Section~\ref{s_uoi}, we use a matrix,~$W_{ij}\in\mbb{R}^{n\times n}$ to retain some design flexibility. The only possible way to cancel the unwanted interference now is via what can be referred to as \emph{post-conditioning}. Each agent,~$i\in\mc{V}$, chooses a post-conditioner,~$R_i\in\mbb{R}^{n\times n}$. As before, we assume~$I_\mc{S} = U_\mc{S}S_\mc{S}V_\mc{S}^\top$ to be the projection matrix for some subspace,~$\mc{S}\subseteq\mbb{R}^n$, and modify the transmission as~$S_\mc{S}\wh{\mt}_k^m$, for some auxiliary state variable,~$\wh{\mt}_k^i\in\mbb{R}^{n}$, to be explicitly defined shortly. The modified protocol is
\begin{eqnarray}\label{cpi_uii}
\wh{\mt}_{k+1}^i = \sum_{j\in\mc{N}_i} W_{ij}R_i\left(S_\mc{S}\wh{\mt}_k^j + \Gamma_{i}\sum_{m\in\mc{V}} a_{ij}^m  S_\mc{S}\wh{\mt}_k^m\right).
\end{eqnarray}
The goal is to design an~$R_i$ such that~$R_i\Gamma_i=\mb{0}_{n\times n}$. Following the earlier approaches, we assume that $\mbox{rank}(\Gamma_i)=\ol{\gamma},\forall i\in\mc{V}$, and $\rank(I_{\mc{S}})=\ul{\gamma}$, such that~$\ol{\gamma} + \ul{\gamma} = n$, with SVDs,~$\Gamma_i=U_iS_iV_i^\top$ and~$I_{\mc{S}}=U_{{{\mc{S}}}}S_{{{\mc{S}}}}V_{{{\mc{S}}}}^\top$, where the singular value matrices are arranged as
\begin{eqnarray}\label{SiSS}
S_i=\left[
\begin{array}{ccc}
S_i^{1:\ol{\gamma}}\\
&\mb{0}_{\ul{\gamma}\times \ul{\gamma}}
\end{array}
\right],\qquad
S_{\mc{S}}=\left[
\begin{array}{ccc}
\mb{0}_{\ol{\gamma}\times \ol{\gamma}}\\
&I_{\ul{\gamma}}
\end{array}
\right].
\end{eqnarray}
The next lemma characterizes the post-conditioner,~$R_i$.
\begin{lem}\label{Rlem}
Let~$\Gamma_i=U_iS_iV_i^\top$ and~$S_\mc{S}$ have the structure of Eq.~\eqref{SiSS}. Given the null-space of~$\Gamma_i^\top$, there exists a rank~$\ul{\gamma}$ post-conditioner,~$R_i$, such that~$R_i\Gamma_i=\mb{0}_{n\times n}$. 
\end{lem}
\begin{proof}
We assume that~$U_i$ is partitioned as~$\left[\begin{array}{ccc}
\ol{U}_i~|~\ul{U}_i
\end{array}
\right]$,
where~$\ol{U}_i\in\mbb{R}^{n\times\ol{\gamma}}$ and~$\ul{U}_i\in\mbb{R}^{n\times\ul{\gamma}}$. Clearly,~$\ul{U}_i$ is the null-space of~$\Gamma_i^\top$. Define 
\begin{eqnarray}
R_i = S_\mc{S} \left[\begin{array}{ccc}
\ol{U}_i^\prime~|~\ul{U}_i^\prime
\end{array}
\right]^\top,
\end{eqnarray}
where~$\ul{U}_i^\prime$ is such that~$\oplus\ul{U}_i^\prime=\oplus\ul{U}_i$, and~$\ol{U}_i^\prime$ is arbitrary. By definition, we have~$\ul{U}_i^\top\ol{U}_i=\mb{0}_{\ul{\gamma}\times\ol{\gamma}}$; hence, by construction,~${\ul{U}_i^\prime}^\top\ol{U}_i=\mb{0}_{\ul{\gamma}\times\ol{\gamma}}$. It can be verified that the post-conditioning results into
\begin{eqnarray*}
R_i\Gamma_i 
=\left[\begin{array}{ccc}
\mb{0}&\mb{0}\\
I_{\ul{\gamma}}{\ul{U}_i^\prime}^\top\ol{U}_iS_i^{1:\ol{\gamma}}&\mb{0}
\end{array}
\right]V_i^\top,
\end{eqnarray*}
and the lemma follows. Note that~$R_i=S_\mc{S}U_i^\top$ is a valid choice but it is not necessary.  
\end{proof}

With the help of Lemma~\ref{Rlem}, Eq.~\eqref{cpi_uii} is now given by
\begin{eqnarray}\label{cpi_uii3}
\wh{\mt}_{k+1}^i &=& \sum_{j\in\mc{N}_i} W_{ij}S_{\mc{S}} \left[\begin{array}{ccc}
\ol{U}_i^\prime~|~\ul{U}_i^\prime
\end{array}
\right]^\top
S_\mc{S}\wh{\mt}_k^j.
\end{eqnarray}
Recall that~$\ul{U}_i^\prime$ is an~$n\times\ul{\gamma}$ matrix whose column-span is the same as the column-span of~$\ul{U}_i$, and the column-span of~$\ul{U}_i$ is the null-space of~$\Gamma_i^\top$. We now denote the lower~$\ul{\gamma}\times\ul{\gamma}$ sub-matrix of~$\ul{U}_i^\prime$ by~$\wh{U}_i$. In order to simply the above iterations, we note that
\begin{eqnarray}\label{uii_P1}
S_{\mc{S}} \left[\begin{array}{ccc}
\ol{U}_i^\prime~|~\ul{U}_i^\prime
\end{array}
\right]^\top S_{\mc{S}} &=&\left[
\begin{array}{ccc}
\mb{0}_{\ol{\gamma}\times \ol{\gamma}}\\
&\wh{U}_i^\top
\end{array}
\right],
\end{eqnarray}
and~$\dim(\ul{U}_i^\prime)=\dim(\ul{U}_i)=n-\ol{\gamma}=\ul{\gamma}$. It is straightforward to show that~$\wh{U}_i^\top$ is always invertible. Based on this discussion, the following lemma establishes the convergence of Eq.~\eqref{cpi_uii}.
\begin{lem}\label{lem_ui}
Let~$\Gamma_i=U_iS_iV_i^\top, \forall i\in\mc{V},$ and some projection matrix,~$I_{\mc{S}}=U_{{{\mc{S}}}}S_{{{\mc{S}}}}V_{{{\mc{S}}}}^\top$, have ranks~$\ol{\gamma}$, and~$\ul{\gamma}\triangleq n-\ol{\gamma}$, respectively~$(0\leq\ul{\gamma}\leq n)$, such that~$S_i$ and~$S_{\mc{S}}$ are arranged as in Eq.~\eqref{SiSS}. When~$R_i$ is chosen according to Lemma~\ref{Rlem}, and for each~$i\in\mc{V}$,~$W_{ij}$ is chosen as
\begin{eqnarray}\label{uii_W}
W_{ij} = w_{ij}\left[
\begin{array}{ccc}
\mb{0}_{\ol{\gamma}\times \ol{\gamma}}\\
&\left(\wh{U}_i^\top\right)^{-1}
\end{array}
\right],
\end{eqnarray}
the protocol in Eq.~\eqref{cpi_uii} recovers the average of the last~$\ul{\gamma}$ components of the initial conditions,~$\wh{\mt}_0^i$.
\end{lem}
\begin{proof}
We note that under the given choice for~$R_i$'s, the interference term is~$\mb{0}_n$, and Eq.~\eqref{cpi_uii} reduces to Eq.~\eqref{cpi_uii3}. Now we use Eqs.~\eqref{uii_P1} and~\eqref{uii_W} in Eq.~\eqref{cpi_uii3} to obtain:
\begin{eqnarray*}
\wh{\mt}_{k+1}^i = \sum_{j\in\mc{N}_i} W_{ij}S_{{{\mc{S}}}}U_i^\top S_{{{\mc{S}}}}\wh{\mt}^j_k
= \sum_{j\in\mc{N}_i}  w_{ij}
\left[
\begin{array}{ccc}
\mb{0}_{\ol{\gamma}\times \ol{\gamma}}\\
&I_{\ul{\gamma}}
\end{array}
\right]
\wh{\mt}^j_k,
\end{eqnarray*}
which in the limit as~$k\ra\infty$ converges to
\begin{eqnarray}
\wh{\mt}_\infty^i &=& \frac{1}{N}\sum_{j=1}^N
\left[
\begin{array}{ccc}
\mb{0}_{\ol{\gamma}\times \ol{\gamma}}\\
&I_{\ul{\gamma}}
\end{array}
\right]\wh{\mt}_0^i,\qquad\forall i\in\mc{V}.
\end{eqnarray}
That~$\wh{U}_i^\top$ is invertible is always true because it is a principal minor of an invertible matrix,~$U_i^\top$.
\end{proof}
\noindent Following is the main result of this section.
\begin{thm}\label{th_cui}
Let~$\Gamma_i$'s,~$R_i$'s, and~$W_{ij}$'s, be chosen according to Lemma~\ref{lem_ui}. The protocol in Eq.~\eqref{cpi_uii} under uniform incoming interference recovers the average in a~$\ul{\gamma}$-dimensional subspace,~$\mc{S}$, of~$\mbb{R}^n$.
\end{thm}
\begin{proof}
Without loss of generality, assume that~$\mc{S}$ has a projection matrix,~$I_{\mc{S}}$, with SVD as defined above. Let~$\wh{\mt}_0^i=V_{\mc{S}}^\top\mt_0^i$ and define~$\wt{\mt}_k^i=U_{\mc{S}}\wh{\mt}_k^i, \forall i\in\mc{V}$. Then, from Lemma~\ref{lem_ui}
\begin{eqnarray*}
\wt{\mt}_\infty^i = U_{\mc{S}}\frac{1}{N}\sum_{j=1}^N
\left[
\begin{array}{ccc}
\mb{0}_{\ol{\gamma}\times \ol{\gamma}}\\
&I_{\ul{\gamma}}
\end{array}
\right]V_{\mc{S}}^\top\mt_0^i=\frac{1}{N}\sum_{j=1}^N
U_{\mc{S}}S_{\mc{S}}V_{\mc{S}}^\top\mt_0^i,
\end{eqnarray*}
$\forall i\in\mc{V}$, and the theorem follows.
\end{proof}
Some remarks are in order to explain the mechanics of Theorem~\ref{th_cui}. Let $I_\mc{S}=U_\mc{S}S_\mc{S}V_\mc{S}^\top$ with $ V_\mc{S} =
\left[
\begin{array}{ccc}
\ol{V}_\mc{S}&|&\ul{V}_\mc{S}
\end{array}\right],$ and
$ U_\mc{S} =
\left[
\begin{array}{ccc}
\ol{U}_\mc{S}&|&\ul{U}_\mc{S}
\end{array}\right],$
where~$\ol{V}_\mc{S}$ is the null space of~$I_\mc{S}$. 

(i) When any agent~$i\in\mc{V}$ receives~$S_\mc{S}\wh{\mt}_0^m$ as an interference, it is canceled via the post-conditioning by~$R_i$, regardless of the transmission,~$S_\mc{S}\wh{\mt}_0^m$:
\begin{eqnarray*}
R_i~\Gamma_i~S_\mc{S}\wh{\mt}_0^m =  
S_{\mc{S}}S_i~~V_i^\top~~S_\mc{S}\wh{\mt}_0^m= \mb{0}_n,
\end{eqnarray*}
because of the structure in the~$S_\mc{S}$ and~$S_i$ from Eq.~\eqref{SiSS}. 

{(ii) It is more interesting to observe the effect on the intended transmission,~$j\ra i$, after the post-conditioning and multiplication with~$W_{ij}$. It is helpful to note that~$S_\mc{S}=S_\mc{S}^\dagger$, and consider the transmission as~$S_\mc{S}^\dagger\wh{\mt}_0^j$ instead of~$S_\mc{S}\wh{\mt}_0^j$:
\begin{eqnarray*}
W_{ij}~R_i~~S_\mc{S}^\dagger\wh{\mt}_0^j &=& W_{ij}~~\underbrace{S_{\mc{S}}U_i^\top}_{\scriptsize\mbox{Rx}}~~\underbrace{S_\mc{S}^\dagger\wh{\mt}_0^j}_{\scriptsize\mbox{Tx}}.
\end{eqnarray*}
The operation,~$S_\mc{S}U_i^\top$, by the receiver, Rx, is vital to cancel the interference as shown in the previous step. However, this measure by the receiver also `distorts' the intended transmission. What agent~$i$ receives is now multiplied by a low-rank matrix,~$S_\mc{S}^\dagger$, in general. Consider for a moment that agent~$j$ were to send~$\wh{\mt}_0^j$ and agent~$i$ obtains~$S_\mc{S}U_i^\top\wh{\mt}_0^j$, after the interference canceling operation. How can agent~$i$ choose an appropriate~$W_{ij}$ to undo this post-conditioning? Such a procedure is not possible unless in trivial scenarios, e.g., when the interference was a diagonal matrix and~$U_i=I_n$. \emph{However, the transmitter may preemptively undo the distortion eventually incurred by the receiver's interference canceling operation}. This is precisely what is achieved by sending~$S_\mc{S}^\dagger\wt{\mt}_0^j$. 

(iii) As we discussed, a preemptive measure, sending~$S_\mc{S}^\dagger\wt{\mt}_0^j$, by the transmitter is vital so that the distortion bound to be added at the receiver is reversed. This reorientation, however, can be harmful, e.g.,~$\wh{\mt}_0^j$ may only contain meaningful (non-zero) information in the first~$\ol{\gamma}$ components and the multiplication by~$S_\mc{S}$ destroys this information. To avoid this issue, we choose the initial condition at each agent as~$\wh{\mt}_0^i=V_{\mc{S}}^\top\mt_0^i$; the first transmission at any agent~$i$ is thus:
	\begin{eqnarray*}
    S_\mc{S}\wh{\mt}_0^i &=&S_\mc{S}V_\mc{S}^\top\mt_0^i = 
    \left[
    \begin{array}{cc}
    \mb{0}_{\ol{\gamma}}\\ 
    \ul{V}_\mc{S}^\top\mt_0^i
    \end{array}
    \right],
    \end{eqnarray*}
    which is to transform any arbitrary initial condition \emph{orthogonal to the null-space} of the desired signal subspace,~$\mc{S}$. Since, the signal subspace,~$\mc{S}$, is~$\ul{\gamma}$-dimensional, retaining only the last~$\ul{\gamma}$ components, after the transformation by~$V_\mc{S}^\top$, suffices.

(iv) We choose~$W_{ij}$ according to Eq.~\eqref{uii_W} and obtain
\begin{eqnarray*}
\wh{\mt}_1^i = \sum_{j\in\mc{N}_i}W_{ij}R_iS_\mc{S}\wh{\mt}_0^j
= 
S_\mc{S}V_\mc{S}^\top\sum_{j\in\mc{N}_i}w_{ij}\mt_0^j = S_\mc{S}V_\mc{S}^\top \mb{x}_{1}^i, 
\end{eqnarray*}
$\forall i\in\mc{V}$, where~$\mb{x}_{k}^i$ are the interference-free consensus iterates. Now lets look at~$\wh{\mt}_2^i$, ignoring the interference terms as they are~$\mb{0}_n$, regardless of the transmission:
\begin{eqnarray*}
\wh{\mt}_{2}^i = \sum_{j\in\mc{N}_i} W_{ij}R_iS_\mc{S}~~S_\mc{S}V_\mc{S}^\top \mb{x}_1^j
=S_\mc{S}V_\mc{S}^\top \mb{x}_2^i,
\end{eqnarray*}
by the same procedure that we followed to obtain~$\wh{\mt}_1^i$. In fact, the process continues and we get $\wh{\mt}_{k+1}^i = S_\mc{S}V_\mc{S}^\top \mb{x}_{k+1}^i,$ or $\wh{\mt}_\infty^i = S_\mc{S}V_\mc{S}^\top \mb{x}_\infty^i$, and the average in~$\mc{S}$, is obtained by $\wt{\mt}_\infty^i = U_\mc{S}\wh{\mt}_\infty^i = U_\mc{S}S_\mc{S}V_\mc{S}^\top \mb{x}_\infty^i = I_\mc{S}\mt_\infty^i.$

\subsection{Illustration of Theorem~\ref{th_cui}}
We now provide a graphical illustration of Theorem~\ref{th_cui}. The network is comprised of~$N=10$ agents each with a randomly chosen initial condition on a~$2$-dimensional subspace,~$\mc{S}$, of~$\mbb{R}^3$, shown in Fig.~\ref{f11_uii} (a). Incoming interference is chosen randomly as a one-dimensional subspace at each agent, shown as grey lines in Fig.~\ref{f11_uii} (b). It can be easily verified that the span of all of the interferences,~$\oplus_{i\in\mc{V}}\Gamma_i$, is the entire~$\mbb{R}^3$. The initial conditions are now transformed with~$V_\mc{S}^\top$ so that the transmission,~$S_\mc{S}\wh{\mt}_k^i$, does not destroy the signal subspace,~$\mc{S}$. This transformation is shown in Fig.~\ref{f11_uii} (c). Consensus iterations are implemented in this transformed subspace,~$\wh{\mt}_k^i$, Fig.~\ref{f11_uii} (d), and finally, the iterations,~$\wt{\mt}_k^i$, in the signal subspace,~$\mc{S}$, are obtained via a post-multiplication by~$U_\mc{S}$.   
\begin{figure*}
\centering
\subfigure{\includegraphics[width=1.25in]{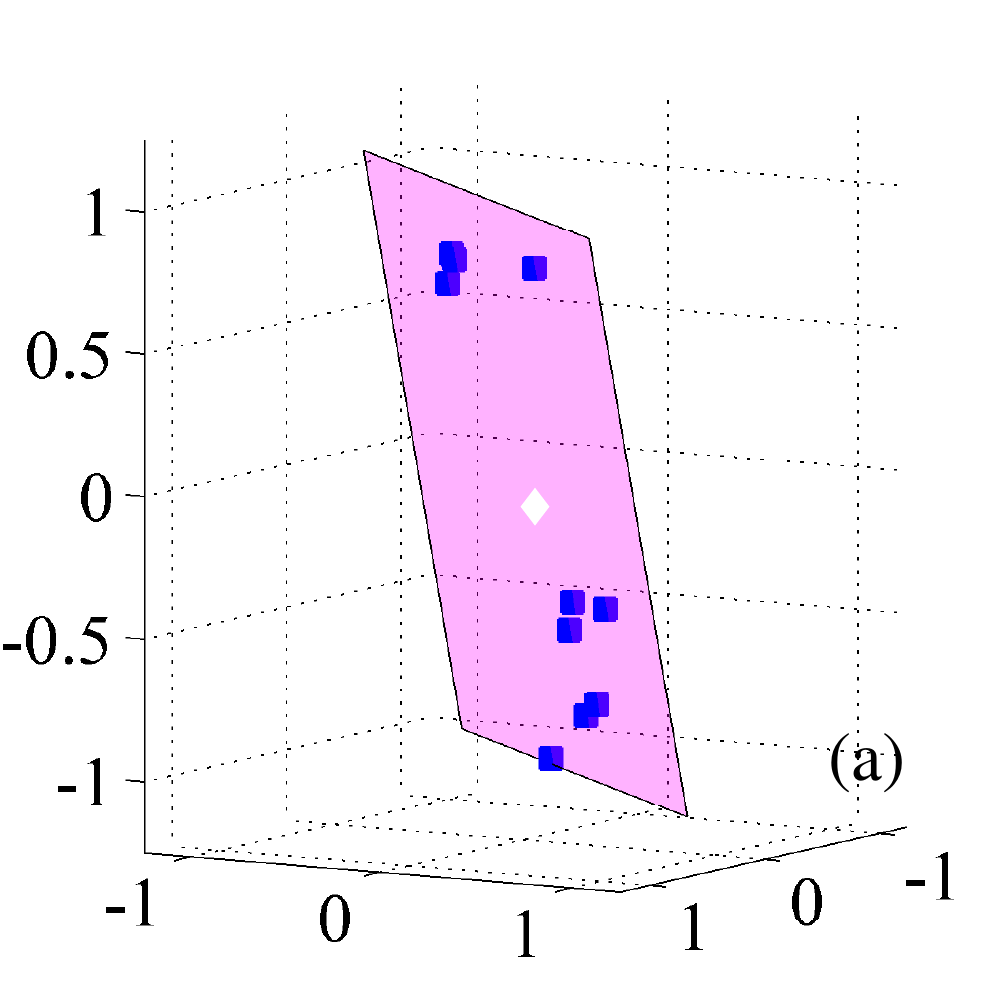}}
\subfigure{\includegraphics[width=1.25in]{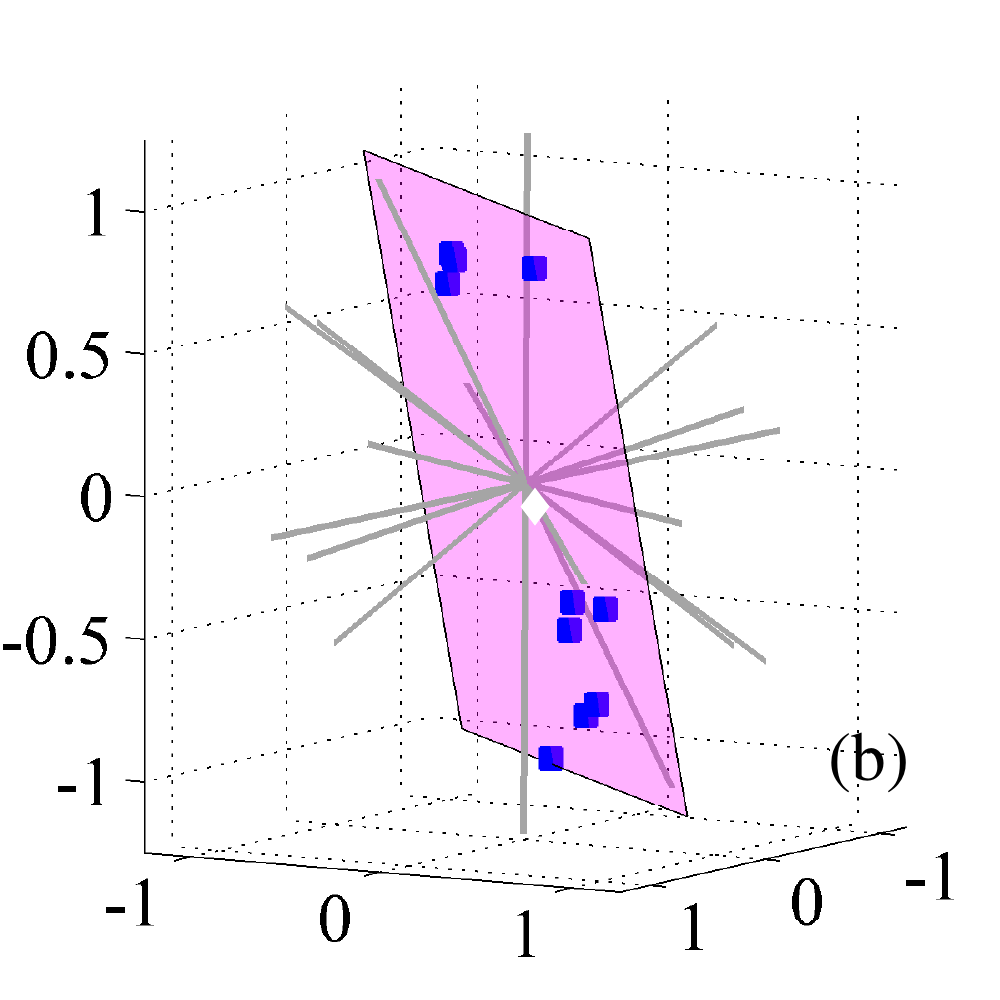}}
\subfigure{\includegraphics[width=1.25in]{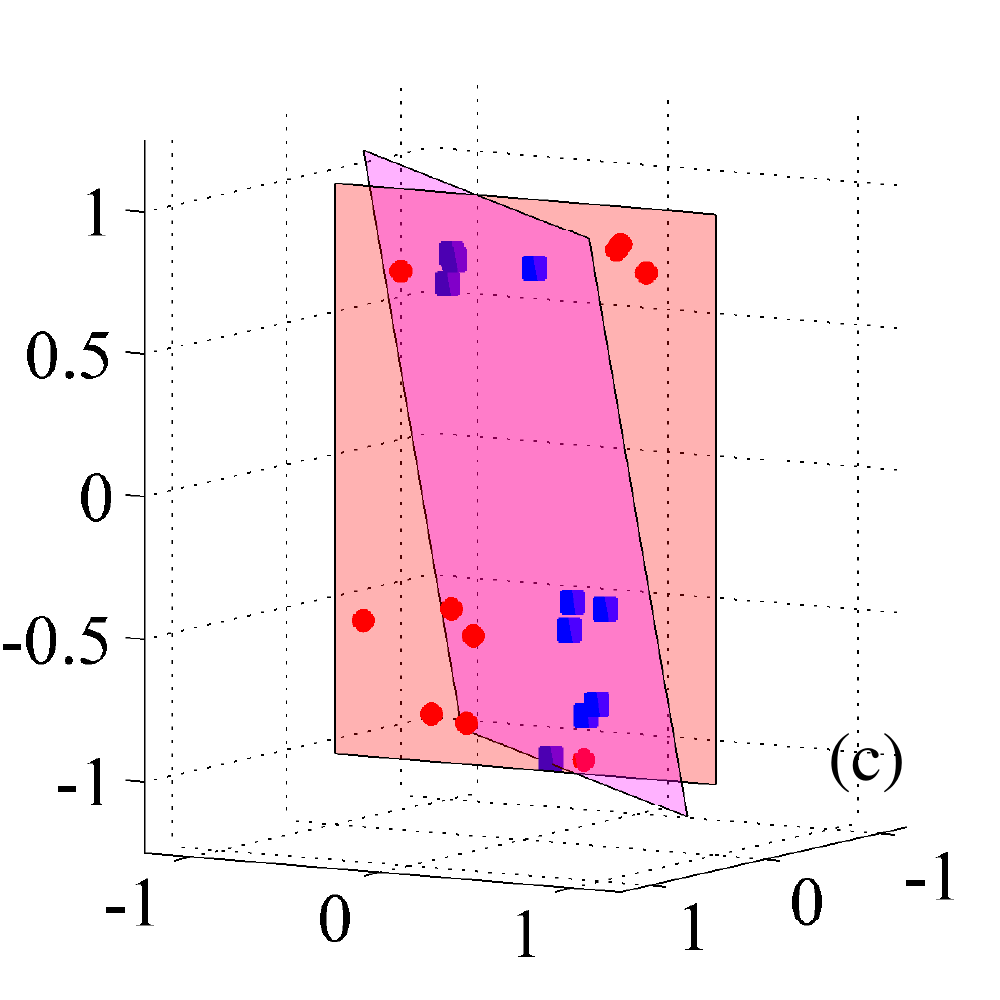}}
\subfigure{\includegraphics[width=1.25in]{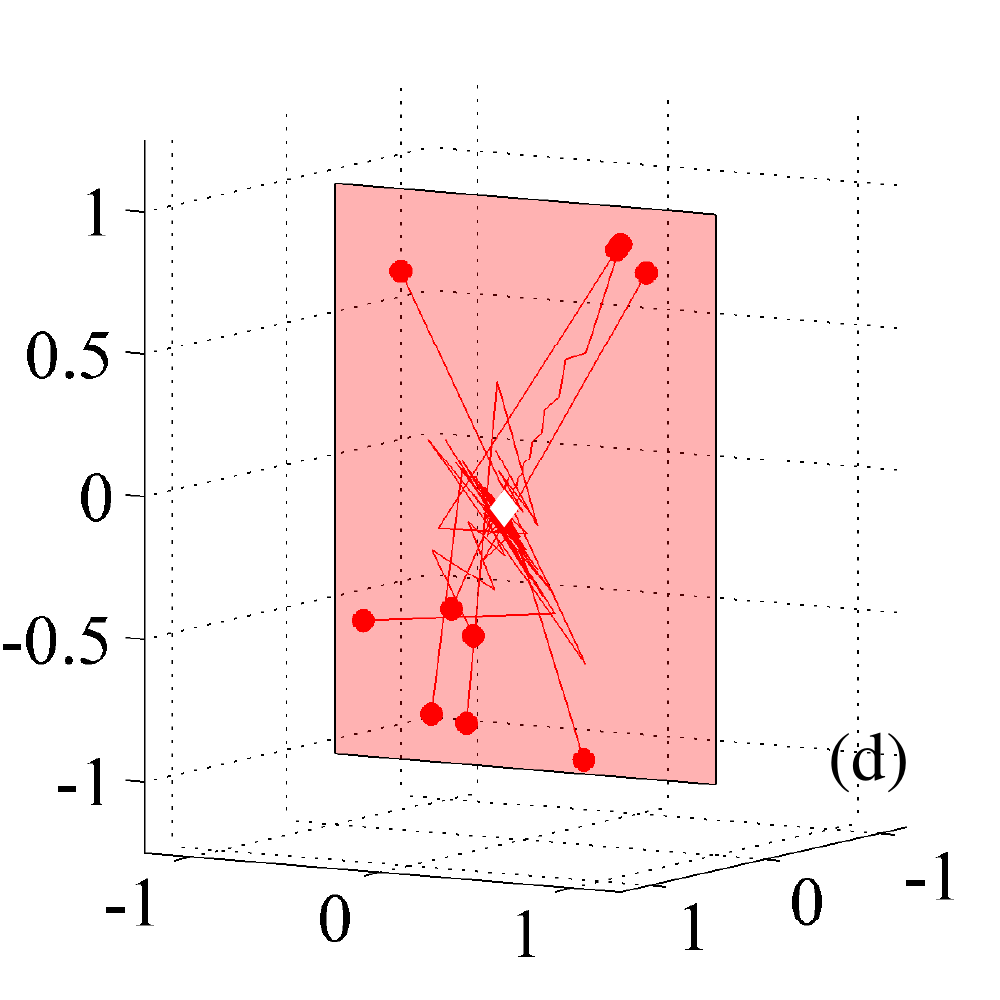}}
\subfigure{\includegraphics[width=1.25in]{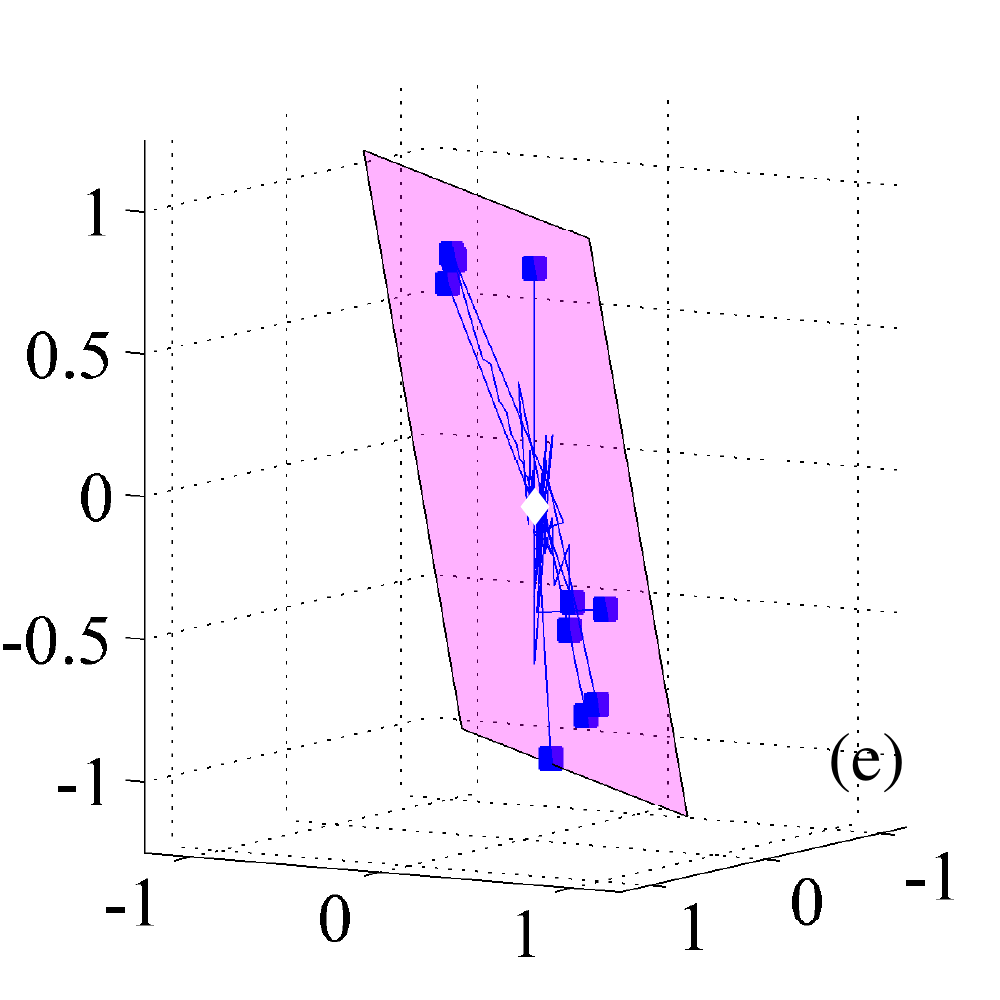}}
\caption{Uniform Incoming Interference: (a) Signal subspace,~$\mc{S}\subseteq\mbb{R}^3$, with~$\dim(\mc{S})=2$. The initial conditions are shown as blue squares and the true average is shown as a white diamond; (b) One-dimensional interference null-spaces at each agent,~$i\in\mc{V}$; (c) Auxiliary state variables,~$\wt{\mt}_0^j=V_\mc{S}^\top\mt_0^i$, shown as red circles; (d) Consensus iterates in the auxiliary states and the average in the auxiliary initial conditions; and, (e) Recovery via~$\wt{\mt}_k^i=U_\mc{S}\wh{\mt}_k^i$.} 
\label{f11_uii}
\end{figure*}

\section{Discussion}
\label{s_discuss}
We now recapitulate the development in this paper. 

{\bf Assumptions:} The exposition is based on three assumptions, (a) and (b) in Section~\ref{pf}, and (c) in Section~\ref{ui_gen}. Assumption (a), in general, ensures that the setup remains practically relevant, and further makes the averaging problem non-trivial. Assumption (b) is primarily for the sake of simplicity; the strategies described in this paper are applicable to the time-varying case. What is required is that when any incoming (or outgoing) interference subspace changes with time, this change is known to the interferer (or the receiver) so that appropriate pre- (or post-) conditioning is implemented. Finally, Assumption (c) is noted to cast a concrete structure on the proposed interference modeling. In fact, one can easily frame the incoming or outgoing interference as a special case of the general framework. However, explicitly noting it establishes a clear distinction among the different structures.

{\bf Conservative Paradigm:} We consider a special case when each of the interference block in the network, see Fig.~\ref{fig_gl}, is identical. This approach, rather restrictive, sheds light on the information alignment notion that keeps recurring throughout the development, i.e. hide the information in the null space of the interference. When the local interferences,~$\Gamma_{ij}^m$, are not identical, we provide a conservative solution that utilizes an interference `blanket' (that covers each local interference subspace) to implement the information alignment. However, as we discussed, this interference blanket soon loses relevance as it may be~$n$-dimensional to provide an appropriate cover. When this is true, the only reliable data hiding is via a zero-dimensional hole (origin) and no meaningful information is transmitted. This conservative approach is improved in the cases of uniform outgoing and incoming interference models. 

{\bf Uniform Outgoing Interference:} The fundamental concept in the uniform outgoing setting is to hide the desired signal in the null-space of the interferences,~$\Gamma_m$'s. This alignment is possible at each transmitter as the eventual interference is only a function of the transmitter. 

{\bf Uniform Incoming Interference:} The basic idea here is to hide the desired signal in the null-space of the transpose of incoming interferences,~$\Gamma_i^\top$'s. This alignment is possible at each receiver as the eventual interference is only a function of the receiver. It can be easily verified that the resulting procedure is non-trivial. 

{\bf Null-spaces}: Incoming and outgoing interference comprise the two major results in this paper. It is noteworthy that both of these results only assume the knowledge of the corresponding interference null-spaces; the basis vectors of these null spaces can be arbitrary while the knowledge of the interference singular values is also not required. It is noteworthy that in a time-varying scenario where the basis vectors of the corresponding null-spaces change such that their span remains the same, no time adjustment is required. 

{\bf Uniform Link Interference}: One may also consider the case when $\Gamma_{ij}^m=\Gamma_{ij}$, see Eq.~\eqref{cpi2}, i.e., each interference gain is only a function of the communication link, $j\ra i$. Subsequently, when each receiving agent, $i\in\mc{V}$, knows the null space of $\Gamma_{ij}^\top$, a protocol similar to the uniform incoming interference can be developed.

{\bf Performance}: To characterize the steady-state error, denoted by~$\mb{e}_\infty^i$ at an agent~$i$, define~$\mb{e}_\infty^i = \mt_\infty^i - I_\mc{S}\mt_\infty^i$, where~$\mt_\infty^i$ is the true average, Eq.~\eqref{pavg}. Clearly,
\begin{eqnarray}\nonumber
\left(I_\mc{S}\mt_\infty^i\right)^\top \mb{e}_\infty^i = (\mt_\infty^i)^\top I_\mc{S}^\top\left(I_n - I_\mc{S}\right)\mt_\infty^i = 0,\qquad \forall i\in\mc{V},
\end{eqnarray}
i.e. the error is orthogonal to the estimate, or the average obtained is the best estimate in~$\mc{S}\subseteq\mbb{R}^n$ of the perfect average. 

\section{Conclusions}
\label{s_conclude}
In this paper, we consider three particular cases of a general interference structure over a network performing distributed (vector) average-consensus. First, we consider the case of uniform interference when the interference subspace is uniform across all agents. Second, we consider the case when this interference subspace depends only on the interferer (transmitter), referred to as \emph{uniform outgoing interference}. Third, we consider the case when the interference subspace depends only on the receiver, referred to as \emph{uniform incoming interference}. For all of these cases, we show that when the nodes are aware of the complementary subspaces (null spaces) of the corresponding interference, consensus is possible in a low-dimensional subspace whose dimension is complimentary to the largest interference subspace (across all of the agents). For all of these cases, we derive a completely local \emph{information alignment} strategy, followed by local consensus iterations to ensure perfect subspace consensus. We further provide the conditions under which this subspace consensus recovers the exact average. The analytical results are illustrated graphically to describe the setup and the information alignment scheme.

\bibliographystyle{IEEEbib}
\bibliography{interference_refs,refs}
\end{document}